\documentclass[journal]{IEEEtran}
\usepackage[T1]{fontenc}% optional T1 font encoding

\def\argmax{\mathop{\rm \arg\!\max}}

\usepackage{graphicx}
\usepackage[noadjust]{cite}
\usepackage{mcite}
\usepackage{amsfonts,helvet}
\usepackage{fancyhdr}
\usepackage{threeparttable}
\usepackage{epsf,epsfig}
\usepackage{amsthm}
\usepackage{amsmath}
\usepackage{siunitx}
\usepackage{amssymb}
\usepackage{dsfont}
\usepackage{subfigure}
\usepackage{color}
\usepackage{enumerate}
\usepackage{hyperref}
\usepackage{cancel}
\usepackage{bbm}
\usepackage{dsfont}
\usepackage[subnum]{cases}
\usepackage{adjustbox}
\usepackage[linesnumbered,ruled]{algorithm2e}
\usepackage{multicol}
\usepackage[english]{babel}
\usepackage{enumitem}

\newtheorem{theorem}{Theorem}

\newtheorem{corollary}{Corollary}

\newtheorem{lemma}{Lemma}

\newtheorem{remark}{Remark}

\newmuskip\pFqmuskip

\def\ba{{\bf a}}

\def\bg{{\bf g}}
\def\bh{{\bf h}}

\def\bn{{\bf n}}

\def\bq{{\bf q}}
\def\br{{\bf r}}
\def\bs{{\bf s}}

\def\bu{{\bf u}}

\def\bw{{\bf w}}
\def\bx{{\bf x}}
\def\by{{\bf y}}
\def\bz{{\bf z}}

\def\bA{{\bf A}}
\def\bB{{\bf B}}
\def\bC{{\bf C}}

\def\bG{{\bf G}}
\def\bH{{\bf H}}
\def\bI{{\bf I}}

\def\bP{{\bf P}}
\def\bQ{{\bf Q}}
\def\bR{{\bf R}}

\def\bU{{\bf U}}

\def\bW{{\bf W}}

\def\cA{\mbox{$\mathcal{A}$}}

\def\cC{\mbox{$\mathcal{C}$}}

\def\cN{\mbox{$\mathcal{N}$}}

\def\cP{\mbox{$\mathcal{P}$}}
\def\cQ{\mbox{$\mathcal{Q}$}}
\def\cR{\mbox{$\mathcal{R}$}}

\def\cV{\mbox{$\mathcal{V}$}}

\def\bbC{\mbox{$\mathbb{C}$}}

\def\bbE{\mbox{$\mathbb{E}$}}

\def\bbV{\mbox{$\mathbb{V}$}}

\usepackage{array}

\makeatletter
\newcommand{\thickhline}{%
    \noalign {\ifnum 0=`}\fi \hrule height 1pt
    \futurelet \reserved@a \@xhline
}
\newcolumntype{"}{@{\hskip\tabcolsep\vrule width 1pt\hskip\tabcolsep}}
\makeatother

\IEEEoverridecommandlockouts

\title{
 Two-Stage Analog Combining in Hybrid Beamforming Systems with Low-Resolution ADCs
}
\author{
Jinseok Choi, \IEEEmembership{Student Member,~IEEE}, Gilwon Lee, \IEEEmembership{Member,~IEEE}, and Brian L. Evans, \IEEEmembership{Fellow,~IEEE} \thanks{
J. Choi, and B. L. Evans are with the Wireless Networking and Communication Group (WNCG), Department of Electrical and Computer Engineering, The University of Texas at Austin, Austin, TX 78701 USA. (e-mail: jinseokchoi89@utexas.edu, bevans@ece.utexas.edu).
G. Lee is with Intel Corporation, Santa Clara, CA 95054 USA. (e-mail: gilwon.lee30@gmail.com).
The authors at The University of Texas at Austin were supported by gift funding from Huawei Technologies.
}
}

\begin{document}
\maketitle

%%%%%%%%%%%%%%%%%%%%%%%%%%%%%%%%%%%%%%%%%%%%%%%%%%%%%%%%
\begin{abstract}
% Propose 2-stage 
% optimality in large-scale mimo
% Intuition - capture channel gain and spread (reducing quantization error)
% Algorithm - SVD/ARV + DFT
% Performance Analysis - One-stage vs. two-stage
% Simulation
In this paper, we investigate hybrid analog/digital beamforming for multiple-input multiple-output (MIMO) systems with low-resolution analog-to-digital converters (ADCs) for millimeter wave (mmWave) communications.
In the receiver, we propose to split the analog combining subsystem into a channel gain aggregation stage followed by a spreading stage.
%Most exiting hybrid combining approaches focus on optimizing  performance for a limited number of radio frequency (RF) chains without considering the coarse quantization effect at ADCs. 
%
% Considering the coarse quantization effect at ADCs, we propose a two-stage analog combining receiver architecture to optimize mutual information (MI) between the transmitted and quantized signals by effectively managing quantization error.
Both stages use phase shifters.
Our goal is to design the two-stage analog combiner to optimize mutual information (MI) between the transmitted and quantized signals by effectively managing quantization error.
% in the low-resolution ADCs.
To this end, we formulate an unconstrained MI maximization problem without a constant modulus constraint on analog combiners, and derive a two-stage analog combining solution.
% We first formulate an unconstrained mutual information (MI) maximization problem which only imposes a unitary condition on an analog combiner without a constant modulus constraint, and then, derive a two-stage analog combining solution. 
The solution achieves the optimal scaling law with respect to the number of radio frequency  chains and maximizes the MI for homogeneous singular values of a MIMO channel.
We further develop a two-stage analog combining algorithm to implement the derived solution for mmWave channels. 
% Since a normalized discrete Fourier transform matrix satisfies the derived second analog combiner condition and is independent to channels, it can be implemented with passive phase shifters with marginal complexity and cost increase. 
By decoupling channel gain aggregation and spreading functions from the derived solution, the proposed algorithm implements the two functions by using array response vectors and a discrete Fourier transform matrix under the constant modulus constraint on each matrix element.
% Since the structure of the derived solution fits well for the two concatenated analog combining architecture,
% Due to the favorable structure of the derived solution for the two concatenated analog combining architecture, 
Therefore, the proposed algorithm provides a near optimal solution for the unconstrained problem, whereas conventional hybrid approaches offer a near optimal solution only for a constrained problem.
% Consequently, the proposed two-stage analog combining algorithm yields superior performance to that of conventional algorithms. 
%Since the two-stage analog combining structure provides a near optimal solution for the unconstrained MI maximization problem, whereas the conventional hybrid structure offers a near optimal solution only for a constrained problem, the proposed two-stage analog combining algorithm yields superior performance compared to that of conventional algorithms. 
The closed-form approximation of the ergodic rate is derived for the algorithm, showing that a practical digital combiner with two-stage analog combining also achieves the optimal scaling law. 
Simulation results validate the algorithm performance and the derived ergodic rate.
% Thus, the derived ergodic rates can characterize the ergodic rate performance of the proposed two-stage analog combining architecture in terms of the system parameters including quantization resolution.
\end{abstract}

\begin{IEEEkeywords}
Two-stage analog combining structure, low-resolution ADCs, mutual information, ergodic rate.
\end{IEEEkeywords}
%%%%%%%%%%%%%%%%%%%%%%%%%%%%%%%%%%%%%%%%%%%%%%%%%%%%%%%%

%%%%%%%%%%%%%%%%%%%%%%%%%%%%%%%%%%%%%%%%%%%%%%%%%%%%%%%
\section{Introduction}
\label{sec:intro}
%%%%%%%%%%%%%%%%%%%%%%%%%%%%%%%%%%%%%%%%%%%%%%%%%%%%%%%

Millimeter wave communications have emerged as a promising technology for 5G communications \cite{pi2011introduction,rappaport2013millimeter}.
Utilizing multi-gigahertz bandwidth in 30-300 GHz frequency ranges enables cellular networks to achieve an order of magnitude increase in achievable rate \cite{andrews2014will}, and a large number of antennas can be packed into tranceivers with very small antenna spacing by leveraging the very small wavelength.
Due to the large number of radio frequency (RF) chains and power-demanding high-resolution ADCs coupled with high sampling rates, however, the significant power consumption at the receivers becomes one of the primary challenges to resolve. 
In this paper, we consider hybrid MIMO receivers with low-resolution ADCs for mmWave communications to address such a challenge by reducing both the number of RF chains and quantization resolution of ADCs.
We propose a two-stage analog combining receiver architecture to maximize the mutual information by effectively managing quantization error as shown in Fig. \ref{fig:receiver}.

% FIGURE 
\begin{figure}[!t]\centering
	\includegraphics[scale = 0.47]{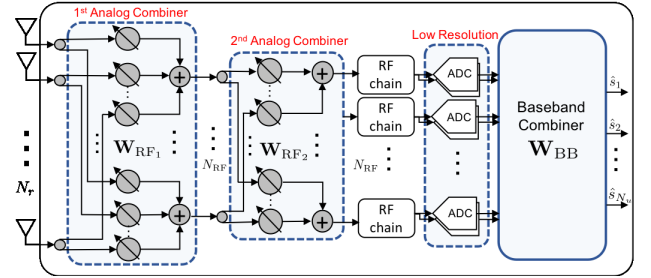}
	\caption{A receiver architecture with two-stage analog combining, low-resolution ADCs and digital combining.} 
	\label{fig:receiver}
\end{figure}

%%%%%%%%%%%%%%%%%%%%%%%%%%%
\subsection{Prior Work}
\label{sec:prior}
%%%%%%%%%%%%%%%%%%%%%%%%%%%

% Hybrid papers for general channels
Hybrid beamforming architectures have been widely investigated to reduce the number of RF chains with minimum communication performance degradation.
%SVD 
Singular value decomposition (SVD)-based analog combining designs were proposed \cite{zhang2005variable,sudarshan2006channel,gholam2011beamforming} as the SVD transceiver maximizes the channel capacity.
In \cite{zhang2005variable}, hybrid precoder and combiner design methods were developed by extracting the phases of the elements of the singular vectors.
Considering correlated channels, the SVD of the MIMO channel covariance matrix was used for analog combiner design to maximize mutual information in \cite{sudarshan2006channel}.
%for a spatial multiplexing system in 
% Rate analysis
The performance of hybrid precoding systems was analyzed for MIMO downlink communications \cite{ying2015hybrid,sohrabi2016hybrid}.
It was shown that hybrid beamforming systems with a small number of RF chains can achieve the performance comparable to fully digital beamforming systems.
%It was shown that with maximum ratio transmission, the asymptotic signal-to-interference-plus-noise ratio (SINR) for the hybrid beamforming system is reduced only by a factor of $\pi/4$ compared to a fully digital beamforming system.
%with one RF chain per user is reduced only by a factor of $\pi/4$ compared to a fully digital beamforming system.
% MMSE
%It was further shown in \cite{sohrabi2016hybrid} that if the number of RF chains is twice the total number of data streams, the hybrid system can achieve the performance of fully digital beamforming system for downlink communications. 
For MIMO uplink communications, the Gram-schmidt based analog combiner design algorithm was developed in \cite{li2016robust} to orthogonalize multiuser signals.

% mmWave channels + ARV + codebook
For mmWave channels, hybrid beamforming techniques were proposed by exploiting the limited scattering of the channels \cite{el2014spatially, alkhateeb2014channel, bogale2014beamforming, rusu2015low, chen2015iterative, liang2014low, alkhateeb2015limited,mendez2016hybrid}.
Adopting array response vectors (ARVs) for analog beamformer design, orthogonal matching pursuit (OMP)-based algorithms were developed in \cite{el2014spatially, alkhateeb2014channel, bogale2014beamforming, rusu2015low, chen2015iterative}.
The proposed OMP-based algorithm in \cite{el2014spatially} approximates the minimum mean squared error (MMSE) combiner with a fewer number of RF chains than the number of antennas by using ARV-based analog combiners.
The OMP-based algorithm in \cite{el2014spatially} was further improved  by combining OMP and local search to reduce the computational complexity \cite{rusu2015low} and by iteratively updating the phases of the phase shifters \cite{chen2015iterative}.
A channel estimation technique was also proposed by using hierarchical multi-resolution codebook-based ARVs  with low training overhead in \cite{alkhateeb2014channel}.
By leveraging the sparse nature of mmWave channels, the proposed algorithms with ARV-based analog beamformers achieved the comparable performance with greatly reduced cost and power consumption compared to fully digital systems.
%Considering limited feedback of channel state information (CSI), an ARV codebook-based analog beamformer design algorithm was also proposed for downlink multiuser mmWave communications \cite{alkhateeb2015limited}.
% Various analog beamforming architectures were examined for mmWave communications in \cite{mendez2016hybrid}, 

% Hybrid + low ADC papers -> design based on exiting hybrid architecture
While the previous studies \cite{zhang2005variable,sudarshan2006channel,gholam2011beamforming,ying2015hybrid,sohrabi2016hybrid, li2016robust, el2014spatially, alkhateeb2014channel, bogale2014beamforming, rusu2015low, chen2015iterative, liang2014low, alkhateeb2015limited,mendez2016hybrid} considered infinite-resolution ADCs in hybrid MIMO systems, hybrid beamforming systems with low-resolution ADCs were investigated in \cite{venkateswaran2010analog,choi2017resolution, choi2018user, mo2017hybrid,  abbas2017millimeter, roth2018comparison} to take advantage of both the hybrid beamforming and low-resolution ADC architectures.
The proposed algorithm in \cite{venkateswaran2010analog} attempted to design an analog combiner by minimizing the MSE including the quantization error.
The analog combiner, however, is not constrained with a constant modulus, and the entire combining matrix needs to be designed for each transmitted symbol separately. 
Without considering the coarse quantization effect in combiner design, bit allocation techniques \cite{choi2017resolution} and user scheduling methods \cite{choi2018user} were developed for a given ARV-based analog combiner. 
In \cite{mo2017hybrid, abbas2017millimeter}, an alternating projection method was adopted to implement SVD-based analog combiners.
The performance analysis of hybrid MIMO systems with low-resolution ADCs in \cite{mo2017hybrid} showed the superior tradeoff between performance and power consumption compared to fully digital systems and hybrid systems with infinite-reoslution ADCs. 
In \cite{roth2018comparison}, a subarray antenna structure was considered, and an ARV-based combining algorithm was used to select the ARV that maximizes the aggregated channel gain.
Although the analysis in \cite{mo2017hybrid, abbas2017millimeter, roth2018comparison} provided useful insights for the hybrid architecture with low-resolution ADCs such as the achievable rate and power tradeoff, the quantization error was not explicitly taken into account in the hybrid beamformer design.
% Quantization needs to be considered
Consequently, considering the coarse quantization effect in the analog combiner design is still an open question.
% the coarse quantization effect in analog combiner design is not considered in most exiting hybrid combining approaches, and 

%%%%%%%%%%%%%%%%%%%%%%%%%%%
\subsection{Contributions}
\label{sec:cont}
%%%%%%%%%%%%%%%%%%%%%%%%%%%

% Propose 2-stage 
% optimality in large-scale mimo
% Intuition - capture channel gain and spread (reducing quantization error)
% Algorithm - SVD/ARV + DFT
% Performance Analysis - One-stage vs. two-stage
% Simulation

In this paper, we derive a near optimal analog combining solution for an unconstrained MI maximization problem in hybrid MIMO systems with low-resolution ADCs.
We, then, propose a two-stage analog combining architecture to properly implement the derived solution under a constant modulus constraint on each phase shifter.
Splitting the solution into a channel gain aggregation stage by using ARVs and a gain spreading stage by using a discrete Fourier transform (DFT) matrix, the two-stage analog combining structure realizes the derived near optimal combining solution with phase shifter-based analog combiners for mmWave communications.
% we investigate hybrid MIMO systems with low-resolution ADCs for mmWave communications to resolve the power consumption problem of mmWave receivers.
% we propose a two-stage analog combining design for hybrid MIMO receivers with low-resolution ADCs and 
The contributions of this paper can be summarized as follows:
\begin{itemize}[leftmargin=*]
\item 
%We consider a hybrid MIMO receiver architecture with low-resolution ADCs.
% We propose a two-stage analog combining receiver architecture to optimize the communication performance by effectively managing quantization error under the constant modulus constraint on the analog combiner. 
Without imposing a constant modulus constraint on an analog combiner, we formulate an unconstrained MI maximization problem for a hybrid MIMO system with low-resolution ADCs.
For a general channel, we derive a near optimal analog combining solution which consists of (1) any semi-unitary matrix that includes the singular vectors of the signal space in the channel matrix and (2) any unitary matrix with constant modulus.
The first and second parts in the derived solution can be considered as a channel gain aggregation function that collects the entire channel gains into the lower dimension and a spreading function that reduces quantization error by spreading the aggregated gains over RF chains, respectively.
We show that the derived solution achieves the optimal scaling law with respect to the number of RF chains and maximizes the MI when the singular values of a MIMO channel are the same.
\item We further develop an ARV-based two-stage analog combining algorithm to implement the derived solution for mmWave channels under the constant modulus constraint on each phase shifter. 
% Considering practical implementation with quantized phase shifters, we adopt an ARV codebook for the proposed algorithm.
Decoupling the channel gain aggregation and spreading functions from the solution, the algorithm implements the aggregation and spreading functions by using ARVs and a DFT matrix without losing the optimality of the solution in the large antenna array regime.
% under the constant modulus constraint
%The channel gain aggregation and spreading functions of the derive solution can be properly implemented by using ARVs and a DFT matrix, respectively.
% shows a favorable structure in implementation for the two concatenated analog combining architecture,
% Due to the favorable structure of the derived solution for the two concatenated analog combining architecture, 
Therefore, the two-stage analog combiner obtained from the proposed algorithm under the constant modulus constraint also provides a near optimal solution for the unconstrained MI maximization problem, whereas conventional hybrid approaches offer a near optimal solution only for a constrained problem.
% Additionally, the derived second analog combiner can be implemented by using a discrete Fourier transform (DFT) matrix, and the DFT matrix is independent of channels. 
Since the DFT matrix is independent of channels, only passive phase shifters need to be appended to a conventional hybrid MIMO architecture with marginal complexity and cost increase, while achieving a large MI gain.
% Consequently, the proposed two-stage analog combining algorithm yields superior performance to that of conventional algorithms. 
%Since the two-stage analog combining structure provides a near optimal solution for the unconstrained MI maximization problem, whereas the conventional hybrid structure offers a near optimal solution only for a constrained problem, the proposed two-stage analog combining algorithm yields superior performance compared to that of conventional algorithms.
\item We derive a closed-form approximation of the ergodic rate with a maximum ratio combining (MRC) digital combiner for the proposed algorithm.
The derived rate characterizes the ergodic rate performance of the proposed two-stage analog combining architecture in terms of the system parameters including quantization resolution.
The derived rate reveals that the ergodic rate of the MRC combiner achieves the same optimal scaling law with the proposed two-stage analog combiner by reducing the quantization error as the number of RF chains increases. 
\end{itemize}
Simulation results demonstrate that the proposed two-stage analog combining algorithm outperforms conventional algorithms and validate the derived ergodic rate.

{\it Notation}: $\bf{A}$ is a matrix and $\bf{a}$ is a column vector. 
$\mathbf{A}^{H}$ and $\mathbf{A}^T$  denote conjugate transpose and transpose. 
$[{\bf A}]_{i,:}$ and $ \mathbf{a}_i$ indicate the $i$th row and column vector of $\bf A$. 
We denote $a_{i,j}$ or $[\bA]_{i,j}$ as the $\{i,j\}$th element of $\bf A$ and $a_{i}$ as the $i$th element of $\bf a$. $\lambda_i\{{\bf A}\}$ denotes the $i$-th largest singular value of ${\bf A}$. 
$\mathcal{CN}(\mu, \sigma^2)$ is the complex Gaussian distribution with mean $\mu$ and variance $\sigma^2$. 
$\mathbb{E}[\cdot]$ and $\bbV[\cdot]$ represent an expectation and variance operators, respectively.
The correlation matrix is denoted as ${\bf R}_{\bf xy} = \mathbb{E}[{\bf x}{\bf y}^H]$.
The diagonal matrix $\rm diag\{\bf A\}$ has $\{a_{i,i}\}$ at its $i$th diagonal entry, and $\rm diag \{\bf a\}$ or ${\rm diag}\{{\bf a}^T\}$ has $\{a_i\}$ at its $i$th diagonal entry. 
${\rm blkdiag}\{\bA_1,\dots,\bA_N\}$ is a block diagonal matrix with diagonal entries $\bA_1,\cdots,\bA_N$.
${\bf I}$ denotes the identity matrix with a proper dimension and we indicate the dimension $N$ by $\bI_N$ if necessary.
%denotes an identity matrix with a proper dimension and 
$\bf 0$ denotes a matrix that has all zeros in its elements with a proper dimension.
$\|\bf A\|$ represents $L_2$ norm. 
$|\cdot|$ indicates an absolute value, cardinality, and determinant for a scalar value $a$, a set $\cA$, and a matrix $\bA$, respectively.
${\rm Tr}\{\cdot\}$ is a trace operator and  $x(N)\sim y(N)$ indicates $\lim_{N\to\infty}\frac{x}{y}=1$.

% and $\measuredangle\{\bA\}$ is the element-wise phase extractor which extracts phases of the elements in $\bA$ to the same dimensional matrix. {\color{red} $x\sim y$ indicates $\lim_{N_{r}\to\infty}\frac{x}{y}=1$.}

%%%%%%%%%%%%%%%%%%%%%%%%%%%%%%%%%%%%%%%%%%%%%%%%%%%%%%%
\section{System Model}
\label{sec:system}
%%%%%%%%%%%%%%%%%%%%%%%%%%%%%%%%%%%%%%%%%%%%%%%%%%%%%%%

We consider single-cell uplink wireless communications in which the BS is equipped with $N_r$ receive antennas and $N_{\rm RF}$ RF chains with $N_{\rm RF} < N_r$.
The antennas are uniform linear arrays (ULA), and 
% the number of RF chains $N_{\rm RF}$ is $N_{\rm RF} \ll N_r$.
each RF chain is followed by a pair of low-resolution ADCs.
We assume that the BS serves $N_u$ users each with a single transmit antenna with $N_u \leq N_{\rm RF}$.

%%%%%%%%%%%%%%%%%%%%%%%%%%%%%
\subsection{Channel Model}
\label{subsec:channel}
%%%%%%%%%%%%%%%%%%%%%%%%%%%%%

% Considering a geometry-based  channel model, 
The channel ${\bf h}_{\gamma,k}$ of user $k$ is assumed to be the sum of the contributions of scatterers that contribute $L_k$ propagation paths to the channel ${\bf h}_{\gamma,k}$ \cite{ertel1998overview}.
For mmWave channels, the number of channel paths $L_k$ is expected to be small due to the limited scattering \cite{rappaport2013millimeter}. 
The discrete-time narrowband channel of user $k$ can be modeled as
\begin{align}
	\label{eq:channel_geo}
	{\bf h}_{\gamma,k} = \frac{1}{\sqrt{\gamma_k}} \bh_{k} = \sqrt{\frac{N_r}{\gamma_k L_k}}\sum_{\ell = 1}^{L_k}g_{\ell,k} {\bf a}(\phi_{\ell,k})
\end{align}
where $\gamma_k$ denotes the pathloss of user $k$, $g_{{\ell},k}$ is the complex gain of the $\ell${th} propagation path of user $k$, and ${\bf a}(\phi_{{\ell},k})$ is the ARV of the receive antennas corresponding to the azimuth AoA of the $\ell$th path of the $k$th user $\phi_{{\ell},k} \in [-\pi/2,\pi/2]$. 
The complex channel gain $g_{{\ell},k}$ follows an independent and identically distributed (i.i.d.) complex Gaussian distribution, $g_{{\ell},k} \overset{i.i.d}{\sim} \mathcal{CN}(0, 1)$.
The ARV ${\bf a}(\theta)$ for the ULA antennas of the BS is given as
\begin{align}
	\nonumber
	% {\bf a}(\theta) = \frac{1}{\sqrt{N_r}}\Big[1,e^{-j \frac{2\pi d}{\lambda}\sin(\theta)},e^{-j \frac{4\pi d}{\lambda}\sin(\theta)},\dots,e^{-j \frac{2(N_r-1)\pi d}{\lambda}\sin(\theta)}\Big]^T
	{\bf a}(\theta) = \frac{1}{\sqrt{N_r}}\Big[1,e^{-j \pi\vartheta},e^{-j 2\pi \vartheta},\dots,e^{-j (N_r-1)\pi\vartheta}\Big]^T 
\end{align}
where the spatial angle $\vartheta = \frac{2d}{\lambda}\sin(\theta)$ is related to the physical AoA $\theta$, $d$ is the distance between antennas, and $\lambda$ is the signal wave length.
We use $\phi$ and  $\theta$ to denote the physical AoAs of a user channel and physical angles of analog combiners, respectively.
We also use  $\varphi$ and $\vartheta$ to denote the spatial angles for $\phi$ and $\theta$,  respectively, where $ \varphi, \vartheta \in [-1,1]$.
% We assume that $\vartheta$ is a constant value in the range of $[-1,1]$ and $\varphi$ is a uniform random variable $\varphi\sim {\rm Unif}[-1,1]$.

%%%%%%%%%%%%%%%%%%%%%%%%%%%%%
\subsection{Signal and Quantization Model}
\label{subsec:signal}
%%%%%%%%%%%%%%%%%%%%%%%%%%%%%

For simplicity, we consider a homogeneous long-term received SNR network\footnote{We remark that the derived analysis in this paper can also be applicable to a heterogeneous long-term received SNR network with minor modification.} where a conventional uplink power control compensates for the pathloss and shadowing effect to achieve the same long-term received SNR target for all users in the cell \cite{simonsson2008uplink, tejaswi2013survey}. 
Let $\bx =\bP\bs$ be the transmitted user signals where $\bP ={\rm diag}\{\sqrt{\rho\, \gamma_1},\dots,\sqrt{\rho\,\gamma_{N_u}}\}$ is the transmit power matrix and $\bs$ is the $N_u \times 1$ transmitted symbol vector from $N_u$ users.
Further, let $\bH_\gamma = \bH\bB$ represent the $N_r \times N_u$ channel matrix where $\bB = {\rm diag}\{\sqrt{1/\gamma_1},\dots,\sqrt{1/\gamma_{N_u}}\}$.
% is the pathloss matrix.
The received baseband analog signal vector is given as
\begin{align}
    \nonumber
	\br = \bH_\gamma \bx + \bn =\bH\bB \bP\bs + \bn  =  \sqrt{\rho}\bH\bs + \bn
\end{align} 
where $\bn$ indicates the $N_r \times 1$ additive white noise vector.
% Each column of the channel matrix, $\bh_k$, corresponds to the channel vector for each user $k \in \{1,\dots, N_u\}$. 
We assume zero mean and unit variance for the user symbols $\bs$ and noise $\bn$.
The noise follows the complex Gaussian distribution $\bn \sim \cC\cN({\bf 0},\bI_{N_r})$ and thus, we consider $\rho$ to be the SNR. 
After the BS receives the signals from users, the signals are combined via two analog combiners as shown in Fig. \ref{fig:receiver}.
Then, the received baseband analog signal vector becomes
\begin{align}
	\nonumber
	\by &= \sqrt{\rho}\bW_{\rm RF_2}^H\bW_{\rm RF_1}^H \bH \bs + \bW_{\rm RF_2}^H\bW_{\rm RF_1}^H\bn  \\ 
	\label{eq:y} 
	& = \sqrt{\rho}\bW_{\rm RF}^H \bH \bs + \bW_{\rm RF}^H\bn
\end{align}
where $\bW_{\rm RF} = \bW_{\rm RF_1}\bW_{\rm RF_2}$ denotes the two-stage analog combiner, $\bW_{\rm RF_1} \in \bbC^{N_{r}\times N_{\rm RF}}$ is the first analog combiner, and $\bW_{\rm RF_2} \in \bbC^{N_{\rm RF}\times N_{\rm RF}}$ is the second analog combiner.
Each real and imaginary part of the combined signal \eqref{eq:y} are quantized at ADCs with $b$ quantization bits.
Assuming a MMSE scalar quantizer and Gaussian signaling $\bs \sim \cC\cN({\bf 0}, \bI_{N_u})$, we adopt an additive quantization noise model (AQNM) \cite{fletcher2007robust} which shows reasonable accuracy in the low to medium SNR ranges \cite{orhan2015low}. 
The AQNM approximates the quantization process in linear form, which is equivalent to the approximation with Bussgang decomposition for low-resolution ADCs \cite{mezghani2012capacity}.
The quantized signal vector is expressed as \cite{fletcher2007robust,mezghani2012capacity}
\begin{align}
	\label{eq:yq}
	\by_{\rm q} &= \cQ(\by) 
% 	= \alpha_b \by + \bq 
	= \alpha_b \sqrt{\rho}\bW_{\rm RF}^H \bH \bs + \alpha_b \bW_{\rm RF}^H \bn + \bq
\end{align}
where $\cQ(\cdot)$ is the element-wise quantizer, the scalar quantization gain is $\alpha_b = 1 -\beta_b$ where $\beta_b = \bbE[|y-y_{\rm q}|^2]/\bbE[|y|^2]$, and $\bq$ denotes the quantization noise vector.
For $b > 5$ quantization bits, $\beta_b$ is approximated as $\beta_b \approx \frac
{{\pi}\sqrt{3}}{2}2^{-2b}$. 
For $b \leq 5$, the values of $\beta_b$ are listed in Table 1 in \cite{fan2015uplink}.
The quantization noise vector $\bq$ is uncorrelated to the quantization input $\by$ and follows the complex Gaussian distribution $\bq \sim \cC\cN({\bf 0}, \bR_{\bq \bq})$, where the covariance matrix is given as \cite{fletcher2007robust}
\begin{align}
	\label{eq:Rqq}
	\bR_{\bq \bq}\! = \!\alpha_b\beta_b{\rm diag}\big\{\rho\bW_{\rm RF}^H \bH \bH^H\bW_{\rm RF}\! +\! \bW_{\rm RF}^H\bW_{\rm RF}\big\}.
\end{align}
Then, a digital combiner $\bW_{\rm BB} \in \bbC^{N_{\rm RF}\times N_{\rm RF}}$ is applied to the quantized signal in \eqref{eq:yq} as
\begin{align}
	\label{eq:z}
	\bz = \alpha_b \sqrt{\rho}\bW_{\rm BB}^H\bW_{\rm RF}^H \bH \bs + \alpha_b \bW_{\rm BB}^H \bW_{\rm RF}^H \bn + \bW_{\rm BB}^H\bq.
\end{align}

% % TABLE
% \begin{table}[!t]
% \centering
% \caption{The Values of $\beta$ for $b$ Quantization Bits  }
% \label{tb:beta}
% \begin{tabular}{ l c c c c c }
%   \thickhline
% %	\hline
% 	$ b$  & 1 & 2 & 3 & 4 & 5\\
%   	\hline
%  	$\beta$   & 0.3634 & 0.1175 & 0.03454 & 0.009497 & 0.002499 \\
%  	\hline
%   \thickhline
% \end{tabular}
% \end{table}

%%%%%%%%%%%%%%%%%%%%%%%%%%%%%

%In this paper, we first focus on a general channel for the analysis in Section \ref{sec:analysis}. 
%Then, in Section \ref{sec:design}, we propose a hybrid combining algorithm for mmWave channels by exploiting the sparse nature of mmWave channels based on the analysis.

%%%%%%%%%%%%%%%%%%%%%%%%%%%%%%%%%%%%%%%%%%%%%%%%%%%%%%%
\section{Optimality of Two-Stage Analog Combining}
\label{sec:analysis}
%%%%%%%%%%%%%%%%%%%%%%%%%%%%%%%%%%%%%%%%%%%%%%%%%%%%%%%

In this section, we provide a near optimal structure for the first and second analog combiners $\bW_{\rm RF_1}, \bW_{\rm RF_2}$ in low-resolution ADC systems for a general channel.
To this end, we first formulate an unconstrained MI maximization problem without a constant modulus condition on the analog combiner $\bW_{\rm RF}$. 
Then, we derive a near optimal solution for the unconstrained problem, which can be splitted into two different functions corresponding to the two-stage analog combiner.
% for large-scale MIMO systems with low-resolution ADCs where the number of receive antennas $N_r$ is sufficiently large, without considering a constant modulus constraint.
% we provide an optimal structure for the first and second analog combiners $\bW_{\rm RF_1}, \bW_{\rm RF_2}$ in low-resolution ADC systems for a general channel.
% Then, we show how to implement the solution with the two-stage analog combiner in Section \ref{sec:design}.

We consider the MI between the transmit symbols $\bs$ and quantized signals $\by_{\rm q}$ under the AQNM model as a measure to maximize. 
The MI is given as
\begin{align}
	\label{eq:MI}
	&\cC(\bW_{\rm RF})\\  &\! =\!  \log_2\! \Big|\bI_{N_{\rm RF}}\!\! +\! \rho\alpha_b^2\big( \alpha_b^2\bW_{\rm RF}^H\!\bW_{\rm RF}\! +\! {\bf R}_{{\bf q}{\bf q}}\big)^{\!-1}\!\bW_{\rm RF}^H{\bf H}{\bf H}^H\bW_{\rm RF}\Big|. \nonumber
\end{align}
%Recall that $\bW_{\rm RF} = \bW_{\rm RF_1}\bW_{\rm RF_2}$ and $\bR_{\bq \bq}$ is given in \eqref{eq:Rqq}. 
% Since analog combining is implemented using analog phase shifters, all elements of $\bW_{\rm RF}$ are constrained to have the equal norm of $1/\sqrt{N_r}$. 
Using \eqref{eq:MI}, we formulate the maximum MI problem by only assuming a semi-unitary constraint on the analog combiner $\bW_{\rm RF}^H\bW_{\rm RF}=\bI_{N_{\rm RF}}$ as in \cite{mo2017hybrid} to keep the effective noise being white Gaussian noise.
% , which simplifies the computation of MI. 
Accordingly, the relaxed MI maximization problem is formulated as 
\begin{align}
	\label{eq:P1}
    \cP1: ~ \bW_{\rm RF}^{\rm opt} = \argmax_{\bW_{\rm RF}} ~ \cC(\bW_{\rm RF}),~ \text{s.t. } \bW_{\rm RF}^H\bW_{\rm RF}=\bI.
\end{align}
% by dropping the constant modulus constraint. 

Under the perfect quantization system where the number of quantization bits is assumed to be infinite, the optimal analog combiner for the problem $\cP1$ is given as the matrix $\bU_{1:N_{\rm RF}}$ that consists of the first $N_{\rm RF}$ left singular vectors of $\bH$. 
The optimal solution $\bW_{\rm RF}^{\rm opt}$ of the problem $\cP1$ with a finite number of quantization bits, however, is still not known.
We first derive an optimal scaling law 
%which an optimal analog combiner $\bW_{\rm RF}^\star$ achieves 
with respect to the number of RF chains $N_{\rm RF}$, and provide a solution that achieves the scaling law.
\begin{theorem}[Optimal scaling law]
\label{thm:optimality_two_stage}
    For fixed $N_{\rm RF}/N_r = \kappa$ with $\kappa \in (0,1)$, the MI with the optimal combiner $\bW_{\rm RF}^{\rm opt}$ for the problem $\cP_1$ scales with $N_{\rm RF}$ as 
    \begin{align}  
       \label{eq:C_opt}
    %   \cC(\bW_{\rm RF}^{\rm opt}) \sim N_u \log_2 \left(1+\frac{\alpha}{1-\alpha}\frac{N_{\rm RF}}{N_u}\right),
        \cC(\bW_{\rm RF}^{\rm opt}) \sim N_u \log_2 N_{\rm RF}
    \end{align}
    and this optimal scaling law can be achieved by using $\bW_{\rm RF}^\star=\bW_{\rm RF_1}^\star \bW_{\rm RF_2}^\star$ such that:
    \begin{itemize}
        \item[$(i)$] $\bW_{\rm RF_1}^\star = [\bU_{1:N_u} ~ \bU_\perp]$, and 
        \item[$(ii)$] $\bW_{\rm RF_2}^\star$ is any $N_{\rm RF} \times N_{\rm RF}$ unitary matrix that satisfies the constant modulus condition on its elements,  
    \end{itemize}  
    where $\bU_{1:N_u}$ is the matrix of the left-singular vectors corresponding to the first $N_u$ largest singular values of $\bH$ and $\bU_\perp$ denotes the matrix of any orthonormal vectors whose column space is orthogonal to that of $\bU_{1:N_u}$.
\end{theorem}
\begin{proof}
    % See Appendix \ref{appx:optimality_two_stage}.
    
Since the optimal solution for $\cP1$  is not known, we first derive an upper bound of $\cC(\bW_{\rm RF})$ and its scaling law with respect to $N_{\rm RF}$.
We, then, show that adopting $\bW_{\rm RF}^\star = \bW_{\rm RF_1}^\star \bW_{\rm RF_2}^\star$, which satisfies the conditions $(i)$ and $(ii)$ in Theorem \ref{thm:optimality_two_stage}, achieves the same scaling law of the upper bound. 
% This implies $\cC(\bW_{\rm RF}^\star)$ also follows the same scaling law. 
% We, then, show that adopting the two-stage analog combiner of $\bW_{\rm RF}=\bU_{1:N_{\rm RF}}\bW_{\rm DFT}$ achieves the same scaling of the upper bound, as $N_{\rm RF} \to \infty$.
    
An arbitrary semi-unitary analog combiner $\bW_{\rm RF}$ can be decomposed into 
\begin{align}\label{eq:WRF_decomposed}
   \bW_{\rm RF} = [\bU_{||}~\bU_{\perp}] \bar{\bW}_{\rm RF},
\end{align}
where $\bU_{||}$ is an $N_r \times m$ matrix composed of $m$ orthonormal basis vectors whose column space is in the subspace of ${\rm Span}(\bu_1,\cdots,\bu_{N_u})$ with $1\le m \le N_u$, $\bU_{\perp}$ is an $N_r \times (N_{\rm RF}-m)$ matrix composed of ($N_{\rm RF}-m$) orthonormal basis vectors whose column space is in the subspace of ${\rm Span}^\perp(\bu_1,\cdots,\bu_{N_u})$, and $\bar{\bW}_{\rm RF}$ is an $N_{\rm RF} \times N_{\rm RF}$ unitary matrix.
Here, $\bu_i$ is the $i$-th left-singular vector of $\bH$. 
Using \eqref{eq:WRF_decomposed}, the term $\bW_{\rm RF}^H \bH\bH^H\bW_{\rm RF}$ in \eqref{eq:MI} can be re-written as
    \begin{align}\nonumber
       &\bW_{\rm RF}^H \bH\bH^H\bW_{\rm RF} \\ \nonumber
       &= \bar{\bW}_{\rm RF}^H [\bU_{||}~\bU_{\perp}]^H \bU {\pmb \Lambda} \bU^H [\bU_{||}~\bU_{\perp}]\bar{\bW}_{\rm RF}\\ 
       &=\bar{\bW}_{\rm RF}^H \underbrace{\left[ 
       \begin{matrix}
       {\bU_{||}^H\bU_{1:N_u}\pmb \Lambda_{N_u} \bU_{1:N_u}^H\bU_{||}} & {\bf 0} \\
       {\bf 0} & {\bf 0}
       \end{matrix}\right]}_{\triangleq \bQ} \bar{\bW}_{\rm RF} \label{eq:wHHw_change}
    \end{align}
where ${\pmb \Lambda}={\rm diag}\{\lambda_1,\cdots,\lambda_{N_u},0,\cdots,0\} \in \bbC^{N_r\times N_r}$, 
% whose diagonals are the $N_u$ non-zero singular values of $\bH\bH^H$ and ($N_r-N_u$) zeros, i.e.,  ${\pmb \Lambda}={\rm diag}\{\lambda_1,\cdots,\lambda_{N_u},0,\cdots,0\}$, 
%$\pmb \Lambda ={\rm diag}\{\lambda_1,\dots,\lambda_{N_u},0,\dots,0\}$, 
$\pmb \Lambda_{N_u} ={\rm diag}\{ \lambda_1,\dots,\lambda_{N_u}\}$, $\lambda_i$ is the $i$th largest singular value of $\bH\bH^H$, and  $\bU_{1:N_r}=[\bu_1,\cdots,\bu_{N_r}]$. 
The matrix $\bQ$ has $m$ ranks and can be decomposed into $\bQ=\bU_{\bQ}\bar{\pmb \Lambda}\bU_{\bQ}^H$, where $\bU_{\bQ}$ is the $N_{\rm RF} \times N_{\rm RF}$ matrix consisting of $N_{\rm RF}$ singular vectors of $\bQ$; and $\bar{\pmb \Lambda}={\rm diag}\{\bar{\lambda}_1,\cdots,\bar{\lambda}_m,0,\cdots,0\} \in \bbC^{N_{\rm RF} \times N_{\rm RF}}$.
Here, $\bar{\lambda}_i$ is the $i$th largest singular value of $\bQ$. Since $\bU_{\bQ}$ is unitary, $\bar{\bW}_{\rm RF}$ can be re-expressed as 
\begin{align}\label{eq:basis_change}
    \bar{\bW}_{\rm RF}=\bU_{\bQ}\overline{\bW}_{\rm RF}.
\end{align}
and $\overline{\bW}_{\rm RF}$ is still unitary. 
Substituting \eqref{eq:basis_change} into \eqref{eq:wHHw_change}, we have $\bW_{\rm RF}^H\bH\bH^H\bW_{\rm RF} = \overline{\bW}_{\rm RF}^H\bar{\pmb \Lambda}\overline{\bW}_{\rm RF}$ and the MI in \eqref{eq:MI} becomes
\begin{align} \label{eq:c_wrf_wlambdaw}
    & \cC(\bW_{\rm RF}) \\ \nonumber
%     &\!=\! \log_2\!\Big|\bI \!+\! \rho\alpha_b^2\big(\alpha_b \bI \!+\! \alpha_b\beta_b\rho{\rm diag}\{\overline{\bW}_{\rm RF}^H\bar{\pmb \Lambda}\overline{\bW}_{\rm RF}\}\big)^{-1}\overline{\bW}_{\rm RF}^H\bar{\pmb \Lambda}\overline{\bW}_{\rm RF} \Big| \\
     &\!= \!\log_2\!\left|\bI\! +\! \frac{\alpha_b}{\beta_b} {\rm diag}^{-1}\!\!\left\{\overline{\bW}_{\rm RF}^H\bar{\pmb \Lambda}\overline{\bW}_{\rm RF}\!+\!\frac{1}{\beta_b\rho}\bI \right\}\!\overline{\bW}_{\rm RF}^H\bar{\pmb \Lambda}\overline{\bW}_{\rm RF} \right|. 
    \end{align}
Let $\bG = \overline{\bW}_{\rm RF}^H \bar{\pmb \Lambda}^{1/2}=[\bG_{\rm sub}~ {\bf 0}]$, where $\bG_{\rm sub}$ is the $N_{\rm RF} \times m$ submatrix of $\bG$. Then, the MI can be upper bounded as 
\begin{align} \nonumber
    &\cC(\bW_{\rm RF}) \\ \nonumber
    &= \log_2\left|\bI_{N_{\rm RF}} + \frac{\alpha_b}{\beta_b} \bG^H{\rm diag}^{-1}\left\{\|[\bG]_{i,:}\|^2+\frac{1}{\beta_b\rho}\right\}\bG \right| \\ \nonumber
    &= \log_2\left|\bI_m + \frac{\alpha_b}{\beta_b} \bG_{\rm sub}^H{\rm diag}^{-1}\left\{\|[\bG_{\rm sub}]_{i,:}\|^2+\frac{1}{\beta_b\rho}\right\}\bG_{\rm sub} \right| \\ \nonumber
    &\overset{(a)}{=}\log_2\left|\bI_m + \frac{\alpha_b}{\beta_b} \tilde{\bG}_{\rm sub}^H\tilde{\bG}_{\rm sub} \right| \\ \nonumber
    &=\sum_{i=1}^m \log_2 \left( 1+\frac{\alpha_b}{\beta_b}\lambda_i\{\tilde{\bG}_{\rm sub}^H\tilde{\bG}_{\rm sub}\} \right) \\ \nonumber
    &\overset{(b)}{\le} m \log_2\left(1+\frac{\alpha_b}{\beta_b m}\sum_{i=1}^m \lambda_i\{\tilde{\bG}_{\rm sub}^H\tilde{\bG}_{\rm sub}\}\right) \\
    &\overset{(c)}{=}m \log_2\left(1+\frac{\alpha_b}{\beta_b m}\sum_{i=1}^{N_{\rm RF}} \frac{\|[\bG_{\rm sub}]_{i,:}\|^2}{\|[\bG_{\rm sub}]_{i,:}\|^2+\frac{1}{\beta_b\rho}} \right)  
    \label{eq:upper_MI_bound}
\end{align}
where $(a)$ follows by letting $\tilde{\bG}_{\rm sub}$ be the matrix whose each row $i$ is given as  $i$-th row of $\bG_{\rm sub}$ normalized by $\big(\|[\bG_{\rm sub}]_{i,:}\|^2+\frac{1}{\beta_b\rho} \big)^{1/2}$; $(b)$ comes from Jensen's inequality and the concavity of $\log_2(1+x)$ for $x>0$; and $(c)$ is from 
\begin{align} \nonumber
%    \sum_{i=1}^m \lambda_i\{\tilde{\bG}_{\rm sub}^H\tilde{\bG}_{\rm sub}\} = {\rm Tr}\{\tilde{\bG}_{\rm sub}^H\tilde{\bG}_{\rm sub}\} = {\rm Tr}\{\tilde{\bG}_{\rm sub}\tilde{\bG}_{\rm sub}^H\}=\sum_{i=1}^{N_{\rm RF}} \frac{\|[\bG_{\rm sub}]_{i,:}\|^2}{\|[\bG_{\rm sub}]_{i,:}\|^2+\frac{1}{(1-\alpha)\rho}}.
	\sum_{i=1}^m \! \lambda_i\{\tilde{\bG}_{\rm sub}^H\tilde{\bG}_{\rm sub}\} \!=\! {\rm Tr}\{\tilde{\bG}_{\rm sub}^H\tilde{\bG}_{\rm sub}\}\!=\!\sum_{i=1}^{N_{\rm RF}}\! \frac{\|[\bG_{\rm sub}]_{i,:}\|^2}{\|[\bG_{\rm sub}]_{i,:}\|^2\!+\!\frac{1}{\beta_b\rho}}.
\end{align}
The upper bound of $\cC(\bW_{\rm RF})$ in \eqref{eq:upper_MI_bound} can further be upper bounded by  $m\log_2(1+\frac{\alpha_b N_{\rm RF}}{\beta_b m})$ because $\frac{\|[\bG_{\rm sub}]_{i,:}\|^2}{\|[\bG_{\rm sub}]_{i,:}\|^2+\frac{1}{\beta_b\rho}} < 1$.
Since the derivative of this bound with respect to $m$ is positive for $m>0$ with any given $\alpha_b, N_{\rm RF}>0$, it is maximized when $m=N_{u}$, and thus, it scales as $N_u \log_2 N_{\rm RF}$, as $N_{\rm RF} \to \infty$. 
%Accordingly, $N_u \log_2(1+N_{\rm RF})$ is the optimal scaling with respect to $N_{\rm RF}$, {\color{red} which can be achieved by an optimal analog combiner $\bW_{\rm RF}^\star$ if exits}.

Now, we prove that the scaling law can be achieved by the two-stage analog combiner $\bW_{\rm RF}^\star = \bW_{\rm RF_1}^\star\bW_{\rm RF_2}^\star$ in Theorem \ref{thm:optimality_two_stage}.
%In this case, $\bW_{\rm RF}^{\star H}\bH\bH^H\bW_{\rm RF}^{\star} = \bW_{\rm RF_2}^{\star H}{\pmb \Lambda}_{N_{\rm RF}}\bW_{\rm RF_2}^{\star}$ where ${\pmb \Lambda}_{N_{\rm RF}} = {\rm diag}\{\lambda_1,\cdots,\lambda_{N_u},0,\cdots,0\} \in \bbC^{N_{\rm RF}\times N_{\rm RF}}$.
% is the $N_{\rm RF}\! \times \! N_{\rm RF}$ matrix whose diagonals are the $N_u$ non-zero singular values of $\bH\bH^H$ and ($N_{\rm RF}\!-\!N_u$) zeros.
%, i.e.,  ${\pmb \Lambda}_{N_{\rm RF}}\!=\!{\rm diag}\{\!\lambda_1,\cdots,\lambda_{N_u},0,\cdots,0\!\}$. 
Let $ \bC \triangleq \bW_{\rm RF_2}^{\star H}{\pmb \Lambda}_{N_{\rm RF}}\bW_{\rm RF_2}^{\star}$.
From $\bW_{\rm RF}^{\star H}\bH\bH^H\bW_{\rm RF}^{\star} = \bW_{\rm RF_2}^{\star H}{\pmb \Lambda}_{N_{\rm RF}}\bW_{\rm RF_2}^{\star} = \bC$ where ${\pmb \Lambda}_{N_{\rm RF}} = {\rm diag}\{\lambda_1,\cdots,\lambda_{N_u},0,\cdots,0\} \in \bbC^{N_{\rm RF}\times N_{\rm RF}}$ and \eqref{eq:c_wrf_wlambdaw}, we have 
\begin{align}
	\nonumber
    &\cC(\bW_{\rm RF}^{\star})\\ 
    \label{eq:c_w1w2}
    &= \log_2\left|\bI_{N_{\rm RF}} + \frac{\alpha_b}{\beta_b} {\rm diag}^{-1}\left\{\bC+\tfrac{1}{\beta_b\rho}\bI_{N_{\rm RF}}\right\}\bC \right| \\
    \label{eq:W2_spread}
    &\overset{(a)}{=}\!\log_2\left|\bI \!+\! \frac{\alpha_b}{\beta_b}\!\left(\frac{\sum_{i=1}^{N_u}\lambda_i}{N_{\rm RF}}\!+\!\frac{1}{\beta_b\rho}\right)^{\!\!-1}\!\!\!\!\!\bW_{\rm RF_2}^{\star H}{\pmb \Lambda}_{N_{\rm RF}}\bW_{\rm RF_2}^{\star} \right| \\ 
    \nonumber
    &= \sum_{k=1}^{N_u} \log_2\left(1+\frac{\alpha_b\rho N_{\rm RF} \lambda_k}{N_{\rm RF}+(1-\alpha_b)\rho \sum_{i=1}^{N_u}\lambda_i} \right) \\ 
    \label{eq:achievable_rate_DFT_U}
    &= \sum_{k=1}^{N_u} \log_2\left(1+\frac{\alpha_b\rho N_{\rm RF} \lambda_k/N_r}{\kappa+(1-\alpha_b)\rho \sum_{i=1}^{N_u}\lambda_i/N_r} \right) \\ 
    \nonumber
    &\overset{(b)}{\sim} N_u \log_2 N_{\rm RF},~\text{as}~N_{\rm RF}\to\infty.
\end{align}
Here, $\!(a)$ is from that all diagonal entries of $\bW_{\rm RF_2}^{\star H}{\pmb \Lambda}_{N_{\rm RF}}\!\bW_{\rm RF_2}^{\star}$ are the same as $d_j=\frac{\sum_{i=1}^{N_u}\lambda_i}{N_{\rm RF}}$, for $j=1,\cdots,N_{\rm RF}$ because of the constant modulus property of $\bW_{\rm RF_2}^{\star}$; $(b)$ follows from the fact that as $N_{\rm RF}\to\infty$, i.e., as $N_r \to \infty$, we have  $\frac{1}{N_r}\bH^H\bH \to {\rm diag}\{\frac{1}{L_1}\sum_{\ell=1}^{L_1}|g_{\ell,1}|^2,\cdots, \frac{1}{L_{N_u}}\sum_{\ell=1}^{L_{N_u}}|g_{\ell,N_u}|^2\}$ \cite{ngo2014aspects} by the channel model in \eqref{eq:channel_geo} without the pathloss component and the law of large numbers, which implies
\begin{align} \nonumber
    \frac{\lambda_i}{N_r} \to  \frac{1}{L_i}\sum_{\ell=1}^{L_i}|g_{\ell, i}|^2< \infty,~\text{for}~i=1,\cdots,N_u.
\end{align}
This completes the proof of Theorem \ref{thm:optimality_two_stage}. 
\end{proof}

We note from \eqref{eq:c_w1w2} that $\bW_{\rm RF_1}^\star$ of the two-stage analog combining solution $\bW_{\rm RF}^\star$ aggregates all channel gains into the smaller dimension and provides ($N_{\rm RF} - N_u$) extra dimensions.
Then, as observed in \eqref{eq:W2_spread}, $\bW_{\rm RF_2}^\star$ spreads the aggregated channels gains over all $N_{\rm RF}$ dimensions, which reduces the quantization error by exploiting the extra dimensions.
Accordingly, as the number of RF chains $N_{\rm RF}$ increases, the proposed solution $\bW_{\rm RF}^\star=\bW_{\rm RF_1}^\star\bW_{\rm RF_2}^\star$ achieves the optimal scaling law \eqref{eq:C_opt} by reducing the quantization error. 

\begin{corollary}\label{cor:conventional_sol}
      The conventional optimal solution $\bW_{\rm RF}^{\rm cv}= [\bU_{1:N_{u}} ~ \bU_\perp]$ for perfect quantization systems cannot achieve the optimal scaling law \eqref{eq:C_opt} in coarse quantization systems, and it is upper bounded by
    \begin{align}
        \cC\big(\bW_{\rm RF}^{\rm cv}\big)  <  \cC_{\rm svd}^{\rm ub} = N_u \log_2\left(1+\frac{\alpha_b}{1-\alpha_b}\right). \label{eq:bounded_performance}
    \end{align} 
\end{corollary}
\begin{proof}
     From \eqref{eq:c_w1w2}, we have the following MI by setting $\bW_{\rm RF_2} = \bI$: 
     \begin{align}
        \nonumber
         &\cC\big(\bW_{\rm RF}^{\rm cv}\big)= \log_2\left|\bI+ \frac{\alpha_b}{\beta_b} {\rm diag}^{-1}\left\{{\pmb \Lambda}_{N_{\rm RF}}+\tfrac{1}{\beta_b\rho}\bI\right\}{\pmb \Lambda}_{N_{\rm RF}} \right|\\
         \nonumber
         &= \sum_{i=1}^{N_u} \log_2\left(1+\frac{\alpha_b\lambda_i}{\beta_b\lambda_i + {1}/{\rho}}\right)\stackrel{(a)}< N_u \log_2\left(1+\frac{\alpha_b}{\beta_b}\right).
        % &= \sum_{i=1}^{N_u} \log_2\left(1+\frac{\alpha_b}{\beta_b}\frac{\lambda_i}{\lambda_i + \tfrac{1}{\beta_b\rho}}\right) \stackrel{(a)}< N_u \log_2\left(1+\frac{\alpha_b}{\beta_b}\right)
    \end{align}
    where $(a)$ comes from $\rho > 0 $.
\end{proof} 
Corollary \ref{cor:conventional_sol} shows that the conventional optimal analog combiner $\bW_{\rm RF}^{\rm cv}$ can capture all channel gains but the MI does not scale as that of  $\bW_{\rm RF}^\star=\bW_{\rm RF_1}^\star\bW_{\rm RF_2}^\star$.
% and is bounded by a finite value. 
Since all channel gains after processed through $\bW_{\rm RF}^{\rm cv}$ are concentrated on only $N_u$ RF chains out of $N_{\rm RF}$ RF chains, using $\bW_{\rm RF}^{\rm cv}$ results in severe quantization errors at each of the $N_u$ RF chains. 
Although the channel gains $\{\lambda_i\}$ increase as $N_r$ increases, the quantization errors also increase in proportion to the channel gains for $ \cC\big(\bW_{\rm RF}^{\rm cv}\big)$, yielding only the bounded MI in \eqref{eq:bounded_performance}. 

Again, unlike the conventional solution, the additional second stage analog combiner $\bW_{\rm RF_2}^\star$ proposed in Theorem \ref{thm:optimality_two_stage} spreads the channel gains captured by the first stage combiner $\bW_{\rm RF_1}^\star$ to all $N_{\rm RF}$ RF chains evenly, leading to achieving the optimal scaling law by greatly alleviating quantization errors. 
Intuitively, adopting the second combiner $\bW_{\rm RF_2}^\star$ results in distributing the burden of ADCs confined in few RF chains over all available ADCs of the total RF chains. 
% In particular, as $N_{\rm RF}\to\infty$, the quantization errors in the spread signals at each RF chain become effectively negligible.  
Later, we show that such performance gain from adopting the two-stage analog combining structure can be significant even with a reasonable number of RF chains.  

\begin{theorem}\label{thm:equal_eigenvalue_optimal}
    For the case of homogeneous singular values of $\bH^H\bH$ where all singular values $\{\lambda_i\}$ are equal, the two-stage analog combining solution  $\bW_{\rm RF}^\star =\bW_{\rm RF_1}^\star \bW_{\rm RF_2}^\star$ in Theorem \ref{thm:optimality_two_stage} maximizes the MI  in \eqref{eq:P1} with finite $N_{\rm RF}$, i.e.,  
    % $\lambda=\lambda_1=\cdots=\lambda_{N_u}$,
    % \begin{align}\nonumber
    %     \bW_{\rm RF}^\star =\bW_{\rm RF_1}^\star \bW_{\rm RF_2}^\star
    % \end{align} 
    % that satisfies the two conditions $(i)$ and $(ii)$ of Theorem \ref{thm:optimality_two_stage},
    % is the optimal solution for   
    \begin{gather*}
    	\bW_{\rm RF}^\star = \arg\max_{\bW_{\rm RF}}\cC(\bW_{\rm RF})\\
    	\text{s.t. } \bW_{\rm RF}^H\bW_{\rm RF}=\bI_{N_{\rm RF}} ~\text{and }\lambda_1=\cdots=\lambda_{N_u} = \lambda, 
    \end{gather*}
    and the corresponding optimal MI is given as
    \begin{align}\label{eq:the_optimal_rate_AQNM}
%        \cC_{\rm opt}\triangleq\cC(\bW_{\rm RF}^\star) = N_u \log_2\left(1+\frac{\alpha_b}{1-\alpha_b} \frac{\lambda N_{\rm RF}}{\lambda N_u + \frac{N_{\rm RF}}{(1-\alpha_b)\rho}} \right).
		 \cC_{\rm opt} \! \triangleq \!\cC(\bW_{\rm RF}^\star) \! = \!N_u \!\log_2\!\!\left(\!1\!+\!\frac{\alpha_b\lambda N_{\rm RF}}{\lambda N_u\!(1\!-\!\alpha_b) \!+\!{N_{\rm RF}}/{\rho}} \!\right)\!.
    \end{align}
    % $\bW_{\rm RF}^\star = \arg\max_{\bW_{\rm RF}}\cC(\bW_{\rm RF})$ s.t. $\bW_{\rm RF}^H\bW_{\rm RF}=\bI_{N_{\rm RF}}$ and $\lambda=\lambda_1=\cdots=\lambda_{N_u}$.
    % \begin{align}
    %       \bW_{\rm RF}^\star =\argmax_{\bW_{\rm RF}} ~ \cC(\bW_{\rm RF}),\quad \text{s.t.} \quad \bW_{\rm RF}^H\bW_{\rm RF}=\bI, \lambda_i = \lambda, \forall i.
    % \end{align}
    % $\bW_{\rm RF}^\star =\bU_{1:N_{\rm RF}}\bW_{\rm DFT}$ maximizes $\cC(\bW_{\rm RF})$ when $\lambda=\lambda_1=\cdots=\lambda_{N_u}$.
\end{theorem}
\begin{proof}
    % See Appendix \ref{appx:equal_eigenvalue_optimal}.
    Recall $\bG=\overline{\bW}_{\rm RF}^H\bar{\pmb \Lambda}^{1/2}=[\bG_{\rm sub}~{\bf 0}]$ in the proof of Theorem \ref{thm:optimality_two_stage}, where $\bG_{\rm sub}$ is the $N_{\rm RF} \times m$ submatrix of $\bG$ and $\bar{\pmb \Lambda}={\rm diag}\{\bar{\lambda}_1,\cdots,\bar{\lambda}_m,0,\cdots,0\}$ is the diagonal matrix composed of the singular values of $\bQ$, defined in \eqref{eq:wHHw_change}. 
    From the assumption of $\lambda_1=\cdots=\lambda_{N_u} = \lambda $,  we have 
    \begin{align}\nonumber
%        \lambda_{\rm max}\{\bQ\} &=  
		\max_{\bx\in\mathbb{C}^{N_{\rm RF}}:\|\bx\|=1}\bx^H \bQ \bx &= \max_{\by\in\mathbb{C}^{m}:\|\by\|=1}\lambda \|\bU_{1:N_u}^H\bU_{||}\by\|^2 \\ \nonumber
        &\overset{(a)}{\le} \max_{\by\in\mathbb{C}^{m}:\|\by\|=1}\lambda \|\bU_{1:N_u}^H\|^2\|\bU_{||}\|^2\|\by\|^2 \\ \nonumber
        &=\lambda,
    \end{align}
    where 
%    $ \lambda_{\rm max}\{\bQ\} $ denotes the largest singular value of $\bQ$, 
	$(a)$ comes from the sub-multiplicativity of the norm, and the last equality holds by $\|\bU_{1:N_u}^H\|=1$ and $\|\bU_{||}\|=1$. This implies the singular values of $\bQ$ are bounded as $\bar{\lambda}_i \le \lambda$ for $i=1,\cdots,m$. Hence, 
% 	the square of the norm of each row of $\bG_{\rm sub}$, i.e., 
    $\|[\bG_{\rm sub}]_{j,:}\|^2$ is maximized for any given $\overline{\bW}_{\rm RF}$ when $\bar{\lambda}_i$ achieves $\lambda$ for all $i=1,\cdots,m$. 
    
    We consider the upper bound of $\cC(\bW_{\rm RF})$ in \eqref{eq:upper_MI_bound} and define  
    \begin{align} \nonumber
        \bG^\star_{\rm sub}=\overline{\bW}_{\rm RF}^H\left[
        \begin{matrix}
        \sqrt{\lambda} \bI_m \\ 
        {\bf 0}
        \end{matrix} 
        \right].
    \end{align}
    % Since \eqref{eq:upper_MI_bound} increases as $\|[\bG_{\rm sub}]_{j,:}\|^2$ increases, 
    Then, \eqref{eq:upper_MI_bound} is further upper bounded as 
    \begin{align} \nonumber
        \cC(\bW_{\rm RF}) &\le m \log_2\left(1+\frac{\alpha_b}{\beta_b m}\sum_{i=1}^{N_{\rm RF}} \frac{\|[\bG^\star_{\rm sub}]_{i,:}\|^2}{\|[\bG^\star_{\rm sub}]_{i,:}\|^2+\frac{1}{\beta_b\rho}} \right) \\ \nonumber
        &\overset{(a)}{\le} m \log_2\left(\!1+\frac{\alpha_b N_{\rm RF}}{\beta_b m} \frac{\sum_{i=1}^{N_{\rm RF}}\|[\bG^\star_{\rm sub}]_{i,:}\|^2}{\sum_{i=1}^{N_{\rm RF}}\|[\bG^\star_{\rm sub}]_{i,:}\|^2+\frac{N_{\rm RF}}{\beta_b\rho}} \! \right) \\
        &\overset{(b)}{=}m\log_2\left(1+ \frac{\alpha_b \lambda N_{\rm RF} }{\lambda m {\beta_b}+ {N_{\rm RF}}/{\rho}} \right), \label{eq:upper_MI_equal_eigen}
    \end{align}
    where $(a)$ holds by Jensen's inequality and the concavity of $\frac{x}{x+1}$ for $x>0$; and $(b)$ comes from $ \sum_{i=1}^{N_{\rm RF}}\|[\bG^\star_{\rm sub}]_{i,:}\|^2 = \|\bG^\star_{\rm sub}\|_F^2 = \lambda m.$
    % \begin{align}\nonumber
    %     $ \sum_{i=1}^{N_{\rm RF}}\|[\bG^\star_{\rm sub}]_{i,:}\|^2 = \|\bG^\star_{\rm sub}\|_F^2 = \lambda m.$
    % \end{align}
     Note that \eqref{eq:upper_MI_equal_eigen} is maximized when $m=N_u$ since the derivative of \eqref{eq:upper_MI_equal_eigen} with respect to $m$ is positive for $m>0$ for any given $\alpha_b,\lambda,\rho,N_{\rm RF}>0$. By substituting $\lambda_1=\cdots=\lambda_{N_u} = \lambda$ into \eqref{eq:achievable_rate_DFT_U}, it can be shown that the upper bound of $\cC(\bW_{\rm RF})$ in \eqref{eq:upper_MI_equal_eigen} with $m=N_u$ can be achieved by adopting $\bW_{\rm RF}^{\star}=\bW_{\rm RF_1}^{\star}\bW_{\rm RF_2}^{\star}$.
     This completes the proof of Theorem \ref{thm:equal_eigenvalue_optimal}. 
\end{proof}

Theorem \ref{thm:equal_eigenvalue_optimal} shows the optimality of the proposed two-stage analog combining solution $\bW_{\rm RF}^\star =\bW_{\rm RF_1}^\star\bW_{\rm RF_2}^\star$ in maximizing the MI for any number of RF chains $N_{\rm RF} \geq N_u$ with homogeneous singular values.
We note that such optimality of $\bW_{\rm RF}^\star$  can be nearly achieved for a fixed number of users in large-scale MIMO systems as shown in Remark \ref{rm:homo_singular}.
\begin{remark}\label{rm:homo_singular}
    From Theorem \ref{thm:equal_eigenvalue_optimal}, the two-stage analog combining solution $\bW_{\rm RF}^{\star}=\bW_{\rm RF_1}^{\star}\bW_{\rm RF_2}^{\star}$ in Theorem \ref{thm:optimality_two_stage} maximizes the MI for $\cP1$ as well as achieves the optimal scaling law \eqref{eq:C_opt} in homogeneous massive MIMO networks with a  large number of antennas $N_r$, where each channel elemen $h_{ij}\overset{i.i.d.}{\sim}\cC\cN(0,1)$.
    % where each element of the channel matrix $\bH$ is modeled as $h_{ij}\overset{i.i.d.}{\sim}\cC\cN(0,1)$. 
    This is because as the number of receive antennas $N_r$ increases, $\frac{1}{N_r}\bH^H\bH \to \bI_{N_u}$, i.e, $\frac{1}{N_r}\lambda_i \to 1$, $\forall i$ \cite{Marzetta10TWC}.
\end{remark}
% {\color{blue} (delete? possible to make misunderstanding) Therefore, the proposed analog combining can not only accomplish the optimal scaling law with respect to $N_{\rm RF}$ but also nearly maximize the MI with a finite $N_{\rm RF}$ in large-scale hybrid MIMO systems with low-resolution ADCs.}

Figure \ref{fig:theorem_2_remark_1} shows the simulation results of the MI of the proposed two-stage analog combiner $\bW_{\rm RF}^\star = \bW_{\rm RF_1}^\star \bW_{\rm RF_2}^\star$ in Theorem \ref{thm:optimality_two_stage} and the conventional analog combiner $\bW_{\rm RF}^{\rm cv}$ in Corollary \ref{cor:conventional_sol} which is optimal for infinite-resolution ADC systems.
Here, we use $\bW_{\rm RF_1}^\star=\bW_{\rm RF}^{\rm cv}=\bU_{1:N_{\rm RF}}$ and $\bW_{\rm RF_2}^\star=\bW_{\rm DFT}$, where $\bW_{\rm DFT}$ is an $N_{\rm RF} \times N_{\rm RF}$ normalized DFT matrix, and consider Rayleigh MIMO channels described in Remark \ref{rm:homo_singular}. 
As shown in Fig.  \ref{fig:theorem_2_remark_1}(a), the MI of the proposed two-stage analog combiner almost achieves the optimal MI $\cC_{\rm opt}$ \eqref{eq:the_optimal_rate_AQNM} in Theorem \ref{thm:equal_eigenvalue_optimal} with $\lambda/N_r = 1$ even in the regime of a finite $N_r$. 
We further note that compared with the conventional one-stage combiner $\bW_{\rm RF}^{\rm cv}$ converging to the upper limit $\cC_{\rm svd}^{ub}$, the MI of the two-stage analog combiner logarithmically increases without a limit as $N_r$ increases with $\kappa \approx 1/3$. This follows the optimal scaling law in Theorem \ref{thm:optimality_two_stage}. 

Fig. \ref{fig:theorem_2_remark_1}(b) shows the MI simulation results with respect to the SNR $\rho$. 
The two-stage combiner $\bW_{\rm RF}^\star= \bW_{\rm RF_1}^\star \bW_{\rm RF_2}^\star$ yields superior MI performance to that of $\bW_{\rm RF}^{\rm cv}$, and the MI of $\bW_{\rm RF}^\star$ converges to  $N_u\log_2\left(1+\frac{\alpha_b N_{\rm RF}}{(1-\alpha_b)N_u}\right)$, which is obtained from $\cC_{\rm opt}$ \eqref{eq:the_optimal_rate_AQNM} with  $\rho\to\infty$. 
Therefore, the MI gap between the upper limits of the two combiners $(\bW_{\rm RF}^\star, \bW_{\rm RF}^{\rm cv})$ is 
\begin{align}
    \Delta \!=\! N_u\!\left(\log_2\!\left(\!1\!+\!\frac{\alpha_b N_{\rm RF}}{(1-\alpha_b)N_u}\!\right)\! -\! \log_2\!\left(\!1\!+\!\frac{\alpha_b}{1-\alpha_b}\!\right)\!\right).
\end{align}
Since $N_{\rm RF} \geq N_u$ is considered in this paper, the proposed two-stage combiner $\bW_{\rm RF}^\star$ always yields the higher upper limit of the MI than the SVD-based one-stage combiner $\bW_{\rm RF}^{\rm cv}$.

% FIGURE %%%%%%%%%%%%%%%%%%%%%%%%%
\begin{figure}[t]
\centering
$\begin{array}{c}
{\resizebox{0.9\columnwidth}{!}
{\includegraphics{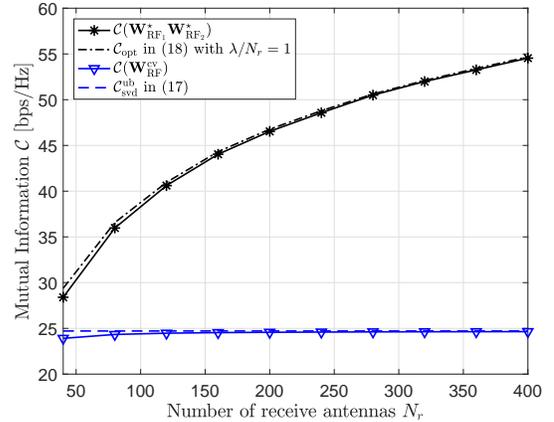}}
}\\ \mbox{\small (a) $\cC$ vs. $N_r$} \\ 
{\resizebox{0.9\columnwidth}{!}
{\includegraphics{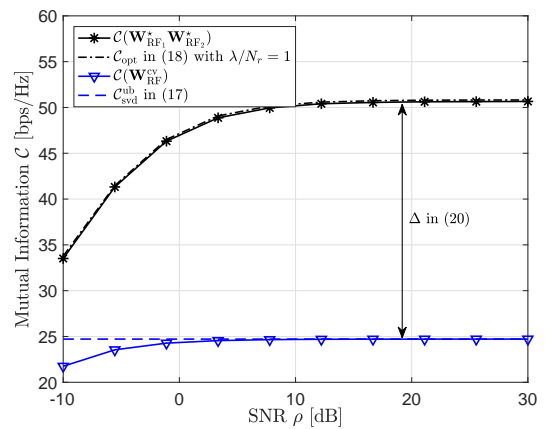}}
}\\ \mbox{\small (b) $\cC$ vs. $\rho$}
\end{array}$
\caption{
The simulation results of the MI with the proposed two-stage analog combining solution $\bW_{\rm RF_1}^{\star}\bW_{\rm RF_2}^{\star}$ and the conventional optimal analog combiner $\bW_{\rm RF}^{\rm cv}$ in the Rayleigh MIMO channels: (a) for $(\rho,N_{\rm RF},N_u,b)=(5~{\rm dB},\lceil \frac{N_r}{3}\rceil,8,2)$ as $N_r$ increases, and (b) for $(N_r,N_{\rm RF},N_u,2)=(256,\lceil \frac{N_r}{3}\rceil,8,2)$ as $\rho$ increases.}
\label{fig:theorem_2_remark_1}
% \vspace{-1 em}
\end{figure}
\section{Two-Stage Analog Combining Algorithm}
\label{sec:design}
%%%%%%%%%%%%%%%%%%%%%%%%%%%%%%%%%%%%%%%%%%%%%%%%%%%%%%%

In the previous section, we derived the analog combining solution for the unconstrained problem $\cP1$.
%without the constant modulus constraint on each entry of the analog combiner. 
However, the constant modulus constraint on each matrix element should be taken into account in designing analog combiners since it is implemented using phase shifters.
% in reducing hardware complexity. 
%Further, an analog phase shifter can control its phase within a pre-defined set of phases with a finite cardinality. 
%This requires codebook-based analog combining design under the constant modulus constraint in practical hybrid beamforming systems. 
We further consider a pre-defined set of phases with a finite cardinality for phase shifters.
Considering channels known at the receiver, we propose a codebook-based two-stage analog combining algorithm for mmWave communications.
%Accrordingly, we propose {\color{red} a codebook-based} two-stage analog combining algorithm to implement the derived analog combiner solution in the previous section by taking into consideration the constant modulus constraint for mmWave communications. 
%Then, we derive the ergodic rate with MRC baseband receivers for the proposed algorithm to analyze the ergodic rate gain from adopting the proposed two-stage analog combining architecture in this section.

%%%%%%%%%%%%%%%%%%%%%%%%%%%%%
\subsection{Proposed Two-Stage Analog Combining Algorithm}
%%%%%%%%%%%%%%%%%%%%%%%%%%%%%

Theorem \ref{thm:optimality_two_stage} provides a practical analog combiner structure that is implementable with a two-stage analog combiner $\bW_{\rm RF} = \bW_{\rm RF_1}\bW_{\rm RF_2}$:
the first analog combiner and the second analog combiner can be considered as a channel gain aggregation matrix and spreading matrix, respectively.
Leveraging such insight and the finding in the following Corollary \ref{cor:AoA_DFT}, 
%we propose two algorithms: SVD-based two-stage analog combining (SVD-TSAC) algorithm for general channels, and 
we propose an ARV-based two-stage analog combining (ARV-TSAC) algorithm for mmWave channels.
\begin{corollary}\label{cor:AoA_DFT}
      When the number of channel paths $L_k$ is limited, the optimal scaling in \eqref{eq:C_opt} can be achieved by using $\tilde\bW_{\rm RF}^\star=\bW_{\rm AoA}\bW_{\rm RF_2}^\star$ as $N_{r} \to \infty$ for fixed $\kappa \in (0,1)$, where $\bW_{\rm AoA}=[\bA_{\rm AoA}, \bA_{\rm AoA}^\perp]$, $\bA_{\rm AoA}=[\ba(\phi_{1,1}),\ba(\phi_{2,1}),\cdots,\ba(\phi_{L_{N_u},N_u})]$, and $\bA_{\rm AoA}^\perp$ is an $N_r \times (N_{\rm RF}-\sum_{k=1}^{N_u}L_k)$ matrix composed of orthonormal basis vectors whose column space is in ${\rm Span}^\perp(\bA_{\rm AoA})$.
\end{corollary}
\begin{proof}
  See Appendix \ref{appx:SVDtoARV}.
\end{proof}
According to Corollary \ref{cor:AoA_DFT}, using ARVs provides a fair tradeoff between practicality in implementaion and performance.
To design the first analog combiner $\bW_{\rm RF_1}$,
%we focus on capturing the most channel gains with a limited number of RF chains by using codebook-based design approach.
we adopt an ARV-codebook based maximum channel gain aggregation approach  
%\cite{roth2018comparison} 
to collect most channel gains into the lower signal dimension by exploiting the sparse nature of mmWave channels. 
We set the codebook of the evenly spaced spatial angles $\cV = \{\vartheta_1,\dots,\vartheta_{|\mathcal{V}|}\}$.
Since selecting $N_{\rm RF}$ ARVs out of the total $|\cV|$ ARVs in the codebook requires $|\mathcal{V}| \choose N_{\rm RF}$ search complexity for the exhaustive method, we propose a greedy-based algorithm to find the best $N_{\rm RF}$ ARVs with greatly reduced complexity.

Algorithm \ref{algo:ARV} describes the proposed ARV-TSAC method.
In Step (a), the ARV with the spatial angle $\vartheta^\star$ which captures the largest channel gain in the remaining channel dimensions $\bH_{\rm rm}$ is selected and it composes a column of the first analog combiner in Step (b). In Step (c), the channel matrix on the remaining dimensions $\bH_{\rm rm}$ is projected onto the subspace of ${\rm Span}^\perp(\ba(\vartheta^\star))$ to remove the channel gain on the space of the selected ARV.
Algorithm \ref{algo:ARV} repeats these steps until $N_{\rm RF}$ ARVs are selected from the codebook $\cV$.
%Based on the analysis in Section \ref{sec:analysis}, 
%we first consider the second analog combiner which is adopted for the two proposed algorithms.
%Then, we consider the first analog combiner design.

% Since Theorem \ref{thm:optimality_two_stage} provides a practical analog combiner structure {\em without losing the optimality in terms of scaling law} that can be implementable.

% as the proposed two-stage analog combining architecture.
% By concatenating a second-stage combiner under the condition $(ii)$ to the conventional first-stage analog combiner as the proposed two-stage analog combining architecture,
\begin{remark}\label{rm:DFT}
    We can implement the second-stage analog combiner that satisfies the condition $(ii)$ of Theorem \ref{thm:optimality_two_stage} by adopting a normalized $N_{\rm RF}\times N_{\rm RF}$ DFT matrix, i.e., $\bW_{\rm RF_2}^\star =\bW_{\rm DFT}$.
    % ,  where $\bW_{\rm DFT}$ is the normalized $N_{\rm RF}\times N_{\rm RF}$ DFT matrix. 
\end{remark}
%Since any phase shifter-based analog combiner can naturally meet the constant modulus condition for the second combiner $\bW_{\rm RF_2}^\star$, we can implement the second combiner $\bW_{\rm RF_2}$  {\em without losing the optimality in terms of scaling law} as shown in Remark \ref{rm:DFT}.
% already meets the constant modulus constraint that is the necessary condition to use phase shifters. 
% According to the optimality condition $(ii)$ in Theorem \ref{thm:optimality_two_stage}, the second analog combiner $\bW_{\rm RF_2}$ can be any unitary matrix whose elements satisfy a constant modulus.
% One representative matrix with such property is a DFT matrix. 
% Thus, we design the second analog combiner as the normalized $N_{\rm RF} \times N_{\rm RF}$ DFT matrix, $\bW_{\rm RF_2} = \bW_{\rm DFT}$.
% $\bullet$ Gradient descent based algorithm
% In our algorithms, we design the second analog combiner as the normalized $N_{\rm RF} \times N_{\rm RF}$ DFT matrix, $\bW_{\rm RF_2} = \bW_{\rm DFT}$.
Employing the DFT matrix for the second analog combiner $\bW_{\rm RF_2}=\bW_{\rm DFT}$ (or any unitary matrix with constant modulus) offers benefits in reducing implementation complexity and power consumption since $\bW_{\rm DFT}$ does not depend on the channel $\bH$ and can be constructed by using passive (or fixed) analog phase shifters.
Accordingly, although the additional $N_{\rm RF}^2$ fully-connected passive phase shifters for the second analog combiner add to the complexity of the proposed architecture in physical area and power consumption, it can be implemented with very low complexity and power consumption in the practical system.
Furthermore, if $N_{\rm RF}$ is a power of two, the fast Fourier transform version of the DFT calculation can be implemented, which reduces the number of additional passive phase shifters to $N_{\rm RF}\log_2 N_{\rm RF}$.
% \cite{XXX}.
% Also, the following remark shows that the conventional optimal solution $\bW_{\rm RF}=\bU_{1:N_{\rm RF}}$ for perfect quantization systems yields bounded rate performance even if $N_{\rm RF} \to \infty$. This implies the second stage analog combiner is essential to achieve the optimal scaling law.

%For mmWave communications, other ARV-based combining design approaches have been developed to approximate an MMSE combiner with infinite-resolution ADCs \cite{el2014spatially}.
%%Ayach, Bogale, Rusu
%In low-resolution ADC systems, however, such approach is not directly applicable since the MMSE combiner is coupled with an analog combiner due to the coarse quantization. 
%Accordingly, we leave the MMSE combiner design approach in hybrid systems with low-resolution ADCs for future work as it is beyond the scope of our paper ({\color{red}Need to doublecheck this}).

%%%%%%%%%%%%%%%%%%%%%%%%%%%%%%%%%%%%%%%%%%%%%%%%%%%%%%%%%%%%%%%%%%%%%%%%%%%%%
\begin{algorithm}[t!]
\label{algo:ARV}
% \vspace{.3em}
 {\bf Initialization}: set $\bW_{\rm RF_1} = $ empty matrix, $\bH_{\rm rm}=\bH$, and $\cV = \{\vartheta_1,\dots, \vartheta_{|\mathcal{V}|}\}$ where $\vartheta_n = \frac{2n}{|\mathcal{V}|}-1$ \\
 \For{$i = 1:N_{\rm RF}$}{
    Maximum channel gain aggregation
 \begin{enumerate}
     \item[(a)] $\ba(\vartheta^\star) = \argmax_{\vartheta \in \mathcal{V}} \|\ba(\vartheta)^H\bH_{\rm rm}\|^2$
     \vspace{0.1em}
     \item[(b)] $\bW_{\rm RF_1} = \big[\ \bW_{\rm RF_1}\ |\  \ba(\vartheta^\star)\ \big] $
     \vspace{0.1em}
     \item[(c)] $\bH_{\rm rm}=\cP_{\ba(\vartheta^\star)}^{\perp}\bH_{\rm rm}$, where $\cP_{\ba(\vartheta)}^{\perp}\!=\!\bI\! -\! \ba(\vartheta)\ba(\vartheta)^H$
     \vspace{0.1em}
     \item[(d)]$\cV = \cV\setminus\{\vartheta^\star\}$
     \vspace{0.1em}
 \end{enumerate}
 }
Set $\bW_{\rm RF_2} = \bW_{\rm DFT}$ where $\bW_{\rm DFT}$ is a normalized $N_{\rm RF} \times N_{\rm RF}$ DFT matrix.\\
\Return{\ }{$\bW_{\rm RF_1}$ and $\bW_{\rm RF_2}$\;}
\caption{ARV-based TSAC}
\end{algorithm}
\subsection{Performance Analysis}
%%%%%%%%%%%%%%%%%%%%%%%%%%%%%

In this subsection, we analyze the ergodic sum rate of the ARV-TSAC algorithm with an MRC baseband combiner.
% for mmWave MIMO communications. 
Once we derive the closed-form ergodic rate, we compare the rate with the one without the second analog combiner $\bW_{\rm RF_2}$ to quantify the ergodic rate gain from employing $\bW_{\rm RF_2}$. 
To this end, we adopt a virtual channel representation \cite{sayeed2002deconstructing} for analytic tractability which captures the sparse property of mmWave channels \cite{mendez2016hybrid,heath2016overview}.
Under the virtual channel representation, the channel vector ${\bf h}_k$ in \eqref{eq:channel_geo} can be modeled as
% EQUATION
\begin{align} 
    \nonumber
     {\bf h}_k = \sqrt{\frac{N_r}{L_k}}\bA \tilde\bg_{k} =\bA\tilde\bh_{{\rm b},k}
    % {\bf h}_k =\sqrt{\frac{N_r}{L_k}}\sum_{i = 1}^{N_r}g_{i,k}\,{\bf a}(\phi_i) = \sqrt{\frac{N_r}{L_k}}\bA \bg_{k} =\bA\bh_{{\rm b},k}
\end{align}
where $\tilde{\bf h}_{{\rm b},k} = \sqrt{\frac{N_r}{L_k}} \tilde\bg_k$ is the $L_k$-sparse beamspace channel of user $k$, i.e., $\tilde{\bf g}_{k}$ has $L_k$ nonzero entries $\stackrel{i.i.d.}\sim \mathcal{CN}(0,1)$, and $\bA = [\ba(\varphi_1), \dots, \ba(\varphi_{N_r})]$ with uniformly spaced spatial angles $\varphi_i$.
Under this representation, we consider the case where the codebook size of Algorithm \ref{algo:ARV} is equal to the number of antennas $|\cV|= N_r$.
Accordingly, the first analog combiner is the $N_r \times N_{\rm RF}$ submatrix of $\bA$ which captures the most channel gain,  $\bW_{\rm RF_1} = \bA_{\rm sub}$.
% We assume that the analog beamformer is composed of the array response vectors corresponding to the $N_{\rm RF}$ largest channel eigenmodes \cite{el2012capacity}, i.e., ${\bf W_{\rm RF}} = {\bf A}_{\rm RF}$
% \in \mathbb{C}^{N_r \times N_{\rm RF}}$ 
% where ${\bf  A}_{\rm RF}$ is a ${N_r \times N_{\rm RF}}$ sub-matrix of $\bf A$.
% (the matrix $\bf \tilde A$ consists of $N_{\rm RF}$ columns of $\bf A$).
We assume that $\bW_{\rm RF_1}$ captures all channel propagation paths from $N_u$ users \cite{kim2015virtual, choi2017resolution}, i.e., $L_k$ channels paths for each user fall within $N_{\rm RF}$ RF chains.
% and consider $L_k = L$, $\forall k$ for simplicity\footnote{Without the equal channel path assumption, the similar results can be derived with minor changes.}.
 % $N_{\rm RF} $ is considered to be large enough to capture all non-zero channel path gains from $N_u$ users.
%Since AoA is the long-term statistics of a channel, 
%Consequently, $\bf A$ includes ${\bf a}(\theta_i)$ for all $L$ dominant AoAs of $N_u$ users. 
%through AoA search training at the BS.
%This analog beamforming allows the BS to see the beamspace received signals at $N_{\rm RF}$ arrival angles. 
% Then, the received signal after the analog beamforming in \eqref{eq:rx_rf_signal} reduces to
% % EQUATION
% \begin{align} 
% \label{eq:y}
% {\bf {y}}  = \sqrt{p_u}{\bf  A}_{\rm RF}^H{\bf Hs} +  {\bf A}_{\rm RF}^H \tilde{\bf n} =\sqrt{p_u} {\bf  H_{\rm b}}{\bf s} + {\bf n}
% \end{align}
% where $  {\bf n} = {\bf A}_{\rm RF}^H\tilde{\bf n} \sim \mathcal{CN}(\mathbf{0}, \mathbf{I}_{N_{\rm RF}})$ as $\bf A$ is unitary. 
% Note that ${\bf H}_{\rm b}$
% % \in \mathbb{C}^{N_{\rm RF} \times N_u}$
% is the ${N_{\rm RF} \times N_u}$ sub-matrix of the beamspace channel matrix $\tilde {\bf H}_{\rm b}$ and contains $\sum_{k=1}^{N_u}L_k$ propagation path gains: 
% \begin{align}
% \label{eq:RFchannel}
% %\nonumber
% {\bf H_{\rm b} = GD}_\gamma^{1/2}
% \end{align}
% where ${\bf G }$\
% % in \mathbb{C}^{N_{\rm RF} \times N_u} $ 
% is the ${N_{\rm RF} \times N_u}$ sub-matrix of the complex gain matrix $\tilde{\bf G}$, corresponding to ${\bf A}_{\rm RF}$.
For simplicity, we further assume $L_k = L$, $\forall k$, in the analysis\footnote{The similar results can be derived with minor changes for general $L_k$.}.
% but it is still solvable with marginal changes.
% Since the beamspace channels $\bh_{{\rm b},k} = \sqrt{\frac{N_r}{L}}\bg_k$ are $L$-sparse vectors, 
Thus, after combining with $\bW_{\rm RF_1}=\bA_{\rm sub}$, the channel becomes $\bH_{\rm b} = \bW_{\rm RF_1}^H\bH$, and the channel vector of user $k$ with the reduced dimension $\bh_{{\rm b},k}\in \bbC^{N_{\rm RF}}$ is 
\begin{align}
	\bh_{{\rm b},k} = \sqrt{\frac{N_r}{L}}\bg_{k}.
\end{align}

We consider $L$ nonzero channel gains to be uniformly distributed within each user channel $\bh_{{\rm b},k}$ and use an indicator function $\mathds{1}_{\{i\in\mathcal{A}\}}$ to characterize the channel sparsity where  $\mathds{1}_{\{i\in \mathcal{A}\}} = 1$ if $i \in \mathcal{A}$, and $\mathds{1}_{\{i\in \mathcal{A}\}} = 0$ otherwise. 
Utilizing $\mathds{1}_{\{\cdot\}}$, we model the $\ell$th complex path gain of user $k$ as
\begin{align}\nonumber
    % \label{eq:g}
    g_{\ell,k} = \xi_{\ell,k}\mathds{1}_{\{\ell \in \mathcal{P}_k\}},\quad \ell = 1,\cdots, N_{\rm RF}, ~ k = 1,\cdots,N_u
\end{align}
where $\xi_{\ell, k} \overset{i.i.d.}\sim \mathcal{CN}(0,1)$, $\forall \ell,k$ and $\mathcal{P}_k = \big\{i\, \big|\, g_{i,k} \neq 0, i = 1,\cdots, N_{\rm RF}\big\}$ is the nonzero index set.
% \begin{align}
%     \nonumber
%      \mathds{1}_{\{i\in \mathcal{A}\}} = 
%     \begin{cases}
%         1 &\quad \text{if } i \in \mathcal{A}\\
%         0 &\quad \text{else.}
%     \end{cases}
% \end{align}  

We consider the MRC combiner $\bW_{\rm BB} = \bar\bH_{\rm b}$ where  $\bar\bH_{\rm b} = \bW_{\rm RF_2}^H\bW_{\rm RF_1}^H\bH$,
% \begin{itemize}
% 	\item For a MRC combiner: $\bW_{\rm BB} = \bar\bH_{\rm b}$ ($L_k = L$ for simplicity, but solvable for general $L_k$)
% % 	\item For a ZF combiner: $\bW_{\rm BB} = \bar \bH_{\rm b}(\bar \bH_{\rm b}^H\bar \bH_{\rm b})^{-1}$
% % 	\item For a MMSE combiner: $\bW_{\rm BB} = \bR_{\by_{\rm q}\by_{\rm q}}^{-1}\bR_{\by_{\rm q}\bx} $
% \end{itemize}
and the received signal $k$ in \eqref{eq:z} becomes
\begin{align}
	\nonumber
	z_k =& \alpha_b\sqrt{\rho}\bar\bh_{{\rm b},k}^H\bar\bh_{{\rm b},k}s_k \\ 
	\label{eq:z_k}
	& +\alpha_b\sqrt{\rho}\sum_{i\neq k}^{N_u}\bar\bh_{{\rm b},k}^H \bar\bh_{{\rm b},i}s_i\! +\! \alpha_b \bar\bh_{{\rm b},k}^H\bW_{\rm RF}^H\bn\! +\! \bar\bh_{{\rm b},k}^H\bq.
\end{align}
From \eqref{eq:z_k}, the achievable rate of the proposed system for the MRC combiner with simplification is given as
\begin{align}
    \label{eq:rk_mrc}
	r_k^{\rm mrc} 
% 	&= \log_2\left(1 + \frac{\rho \alpha\|\bar\bh_{{\rm b},k}\|^4}{\rho\alpha  \sum_{i \neq k}^{N_u}|\bar\bh_{{\rm b},k}^H\bar\bh_{{\rm b},i}|^2 + \|\bar\bh_{{\rm b},k}\|^2 + \rho(1-\alpha)\bar\bh_{{\rm b},k}^H {\rm diag}\big\{\bar\bH_{\rm b}\bar\bH_{\rm b}^H\big\}\bar\bh_{{\rm b},k}}\right)\\
	&\!=\! \log_2\left(\!1\! +\! \frac{\rho \alpha_b\|\bar\bh_{{\rm b},k}\|^4}{\rho\alpha_b  \sum_{i \neq k}^{N_u}|\bar\bh_{{\rm b},k}^H\bar\bh_{{\rm b},i}|^2 + \|\bar\bh_{{\rm b},k}\|^2 + \rho\beta_b\Psi_k}\!\right)
\end{align}
where $	\Psi_k = \bar\bh_{{\rm b},k}^H {\rm diag}\big\{\bar\bH_{\rm b}\bar\bH_{\rm b}^H\big\}\bar\bh_{{\rm b},k}$,
%\begin{align} 
%	\nonumber
%	\Psi_k &= \bar\bh_{{\rm b},k}^H {\rm diag}\big\{\bar\bH_{\rm b}\bar\bH_{\rm b}^H\big\}\bar\bh_{{\rm b},k} \\
%	\nonumber
%	&= \bh_{{\rm b},k}^H\bW_{\rm DFT} {\rm diag}\big\{\bW_{\rm DFT}^H\bH_{\rm b}\bH_{\rm b}^H\bW_{\rm DFT}\big\}\bW_{\rm DFT}^H\bh_{{\rm b},k}
%\end{align}
%$\Psi_k = \bar\bh_{{\rm b},k}^H {\rm diag}\big\{\bar\bH_{\rm b}\bar\bH_{\rm b}^H\big\}\bar\bh_{{\rm b},k} = \bh_{{\rm b},k}^H\bW_{\rm DFT} {\rm diag}\big\{\bW_{\rm DFT}^H\bH_{\rm b}\bH_{\rm b}^H\bW_{\rm DFT}\big\}\bW_{\rm DFT}^H\bh_{{\rm b},k}$.
%Using the achievable rate of user $k$ in \eqref{eq:rk_mrc}, the ergodic rate of user $k$ is given as
and the ergodic rate is
\begin{align}
	\label{eq:E_rk} 
    & \bar r_k^{\rm mrc} = \bbE\Big[r_k^{\rm mrc}\Big] \\
    \nonumber
    &= \!\bbE\!\left[\log_2\!\left(\!1\! +\! \frac{\rho \alpha_b\|\bar\bh_{{\rm b},k}\|^4}{ \rho\alpha_b \sum_{i \neq k}^{N_u}|\bar\bh_{{\rm b},k}^H\bar\bh_{{\rm b},i}|^2 + \|\bar\bh_{{\rm b},k}\|^2 + \rho\beta_b\Psi_k}\!\right)\!\right].
     %   & = \bbE\left[\log_2\left(1 + \frac{\rho \alpha\|\bar\bh_{{\rm b},k}\|^4}{ \rho\alpha \sum_{i \neq k}^{N_u}|\bar\bh_{{\rm b},k}^H\bar\bh_{{\rm b},i}|^2 + \|\bar\bh_{{\rm b},k}\|^2 + \rho(1-\alpha)\bar\bh_{{\rm b},k}^H {\rm diag}\big\{\bar\bH_{\rm b}\bar\bH_{\rm b}^H\big\}\bar\bh_{{\rm b},k}}\right)\right].
\end{align}
% where $r_k^{\rm mrc}$ is shown in \eqref{eq:rk_mrc}.
% Since mmWave channels have a limited number of channel propagation paths, the inter-user interference  $\alpha \rho \sum_{i \neq k}^{N_u}|\bar\bh_{{\rm b},k}^H\bar\bh_{{\rm b},i}|^2$ can be negligible compared to the thermal and quantization noise terms $\|\bar\bh_{{\rm b},k}\|^2 + \rho(1-\alpha)\bar\bh_{{\rm b},k}^H {\rm diag}\big\{\bar\bH_{\rm b}\bar\bH_{\rm b}^H\big\}\bar\bh_{{\rm b},k}$.
% Accordingly, we take a tight upper bound of the ergodic rate \eqref{eq:E_rk} as
% \begin{align}
%     \label{eq:E_rk_upper}
%      \bar r_k^{\rm mrc} &\leq \bbE\left[\log_2\left(1 + \frac{\alpha\rho\|\bar\bh_{{\rm b},k}\|^4}{ 
%      \|\bar\bh_{{\rm b},k}\|^2 + \rho(1-\alpha)\bar\bh_{{\rm b},k}^H {\rm diag}\big\{\bar\bH_{\rm b}\bar\bH_{\rm b}^H\big\}\bar\bh_{{\rm b},k}}\right)\right] \\
%      & \stackrel{(a)}\approx \log_2\left(1 + \frac{\alpha\rho\bbE\big[\|\bar\bh_{{\rm b},k}\|^4\big]}{
%     \bbE\big[ \|\bar\bh_{{\rm b},k}\|^2\big] + \rho(1-\alpha)\bbE\big[\bar\bh_{{\rm b},k}^H {\rm diag}\big\{\bar\bH_{\rm b}\bar\bH_{\rm b}^H\big\}\bar\bh_{{\rm b},k}\big]}\right)
% \end{align}
Since $\bW_{\rm RF_2} = \bW_{\rm DFT}$ is unitary, we have $\|\bar \bh_{{\rm b},i}^H\bar \bh_{{\rm b},j}\| = \| \bh_{{\rm b},i}^H \bh_{{\rm b},j}\|$, $\forall i, j$. 
We approximate the ergodic rate \eqref{eq:E_rk} as
\small
\begin{align} \nonumber
    &\bar r_k^{\rm mrc} \!=\! \bbE\!\left[\log_2\left(\!1 + \frac{\rho \alpha_b\|\bh_{{\rm b},k}\|^4}{ \rho\alpha_b \sum_{i \neq k}^{N_u}|\bh_{{\rm b},k}^H\bh_{{\rm b},i}|^2 + \|\bh_{{\rm b},k}\|^2 + \rho\beta_b{ \Psi}_k}\!\right)\!\right] \\
    \label{eq:E_rk_approx}
    &\!\stackrel{(a)} \approx\!\log_2\!\left(\!1 \!+\! \frac{\rho \alpha_b\bbE\big[\|\bh_{{\rm b},k}\|^4\big]}{ \rho\alpha_b \!\sum_{i \neq k}^{N_u}\!\bbE\big[|\bh_{{\rm b},k}^H\bh_{{\rm b},i}|^2\big] \!+\! \bbE\big[ \|\bh_{{\rm b},k}\|^2\big]\! + \!\rho\beta_b\bbE\big[{\Psi}_k\big]}\right)
\end{align}
\normalsize
where $(a)$ follows from Lemma 1 in \cite{zhang2014power}.
% where
% \begin{align}
%     {\Psi}_k = \bh_{{\rm b},k}^H\bW_{\rm DFT} {\rm diag}\big\{\bW_{\rm DFT}^H\bH_{\rm b}\bH_{\rm b}^H\bW_{\rm DFT} \big\}\bW_{\rm DFT}^H\bh_{{\rm b},k}.
% \end{align}
% where ${ \Psi}_k = \bh_{{\rm b},k}^H\bW_{\rm DFT} {\rm diag}\big\{\bW_{\rm DFT}^H\bH_{\rm b}\bH_{\rm b}^H\bW_{\rm DFT} \big\}\bW_{\rm DFT}^H\bh_{{\rm b},k}$ and 

We first analyze the average quantization error with two-stage analog combining and MRC $\bbE[\Psi_k]$ in \eqref{eq:E_rk_approx}.
Noting that $\Psi_k = \bh_{{\rm b},k}^H\bW_{\rm DFT} {\rm diag}\big\{\bW_{\rm DFT}^H\bH_{\rm b}\bH_{\rm b}^H\bW_{\rm DFT}\big\}\bW_{\rm DFT}^H\bh_{{\rm b},k}$, 
we decompose $\bbE[\Psi_k]$ as $\bbE[\Psi_k] = \bbE[\Psi^{\rm auto}_k] + \bbE[\Psi^{\rm cross}_k]$, and define the auto quantization noise and cross quantization noise variances as
\begin{align}
   	\label{eq:QN_auto}
    &\bbE\big[\Psi_{k}^{\rm auto}\big] \\
    \nonumber
    &= \bbE\Big[\bh_{{\rm b},k}^H\!\bW_{\rm DFT} {\rm diag}\big\{\bW_{\rm DFT}^H\bh_{{\rm b},k}\bh_{{\rm b},k}^H\bW_{\rm DFT}\big\}\bW_{\rm DFT}^H\bh_{{\rm b},k}\Big],\\
   \label{eq:QN_cross}
    &\bbE\big[\Psi_{k}^{\rm cross}\big] \\
    \nonumber
    &\!=\! \bbE\Big[\bh_{{\rm b},k}^H\!\bW_{\rm DFT} {\rm diag}\big\{\!\bW_{\rm DFT}^H\bH_{{\rm b}\setminus k}\bH_{{\rm b}\setminus k}^H\bW_{\rm DFT}\!\big\}\bW_{\rm DFT}^H\bh_{{\rm b},k}\Big]
\end{align}
where $\bH_{{\rm b}\setminus k}$ denotes the channel matrix $\bH_{\rm b}$ without its $k$th column.
Then, \eqref{eq:QN_auto} and \eqref{eq:QN_cross} represent the average quantization errors for the associated user caused by the associated user itself and other users, respectively.
%We derive \eqref{eq:QN_auto} and \eqref{eq:QN_cross} in closed form in the following lemmas.
\begin{lemma}\label{lem:QN_auto}
    For the considered mmWave channel, the auto quantization noise variance for the two-stage analog combining of the ARV-TSAC algorithm with MRC  \eqref{eq:QN_auto} is derived as 
    \begin{align}
        \label{eq:QN_auto_result}
        \bbE\big[\Psi_{k}^{\rm auto}\big] = \frac{2N_r^2}{N_{\rm RF}}.
    \end{align}
\end{lemma}
\begin{proof}
    See Appendix \ref{appx:QN_auto}.
\end{proof}
Note that the quantization noise variance decreases as the number of RF chains $N_{\rm RF}$ increases, which corresponds to the intuition: the second DFT analog combiner spreads the quantization noise over the $N_{\rm RF}$ chains and thus reduces the quantization error more as $N_{\rm RF}$ increases.
\begin{lemma}\label{lem:QN_cross}
     For the considered mmWave channel, the cross quantization noise variance for the two-stage analog combining of the ARV-TSAC algorithm with MRC  \eqref{eq:QN_cross} is derived as
    \begin{align}   
        \label{eq:QN_cross_result}
        \bbE\big[\Psi_{k}^{\rm cross}\big]  = \frac{N_r^2(N_u-1)}{N_{\rm RF}}.
    \end{align}
\end{lemma}
\begin{proof}
    See Apprendix \ref{appx:QN_cross}.
\end{proof}
% {\color{red} In \eqref{eq:QN_cross_result}, $N_{\rm RF}$ comes from the chance not to overlap between the user $k$ and other users.}
% $\bbE[\Psi_{k}^{\rm auto}] = \bbE\big[\bar\bh_{{\rm b},k}^H {\rm diag}\big\{\bar\bh_{{\rm b},k}\bar\bh_{{\rm b},k}^H\big\}\bar\bh_{{\rm b},k}\big]$ and  $\bbE[\Psi_{k}^{\rm cross}] = \bbE\big[\bar\bh_{{\rm b},k}^H {\rm diag}\big\{\bar\bH_{{\rm b}\setminus k}\bar\bH_{{\rm b}\setminus k}^H\big\}\bar\bh_{{\rm b},k}\big]$.
%We note that $\bbE\big[\Psi_{k}^{\rm cross}\big]$ also decreases with $N_{\rm RF}$. 
Since both $\bbE\big[\Psi_{k}^{\rm auto}\big]$ and $\bbE\big[\Psi_{k}^{\rm cross}\big]$ decrease with $N_{\rm RF}$, the quantization error with the proposed two-stage analog combining and MRC combining is expected to decrease as $N_{\rm RF}$ increases, leading the ergodic rate to the same scaling law as in \eqref{eq:C_opt}. 
We derive the approximated ergodic sum rate of \eqref{eq:rk_mrc} in closed form and validate the insight.
\begin{theorem}
    \label{thm:ergodic_rate_mrc}
    For the considered mmWave channel with low-resolution ADCs, the ergodic sum rate of the ARV-based TSAC method with MRC is approximated as
    \begin{align}
        \label{eq:ER_mrc}
        \bar \cR^{\rm mrc} \! \approx \! N_u \! \log_2\!\Bigg(\!1\!+\! \frac{ \rho\alpha_b N_r N_{\rm RF}(1+1/L)}{N_{\rm RF} \!+\! \rho N_r(N_u-1)\! +\! 2\rho(1-\alpha_b) N_r}\!\Bigg).
        % &  \approx N_u \log_2\Bigg(1+ \frac{ \rho\alpha N_r N_{\rm RF}(1+L)}{LN_{\rm RF} + \rho \alpha L N_r (N_u -1) + \rho (1-\alpha)N_r(1+LN_u)}\Bigg).
    \end{align}
\end{theorem}
\begin{proof}
   See Appendix \ref{appx:ergodic_rate_mrc}.
\end{proof}
Note that the derived ergodic rate in \eqref{eq:ER_mrc} is a function of system parameters and provides insights how the ergodic rate is improved with the proposed two-stage analog combining. 
\begin{remark}
\label{rm:scaling_law_MRC}
    Let $\kappa = N_{\rm RF} / N_r $ where $\kappa\in(0,1)$ is a constant value.
    Then, \eqref{eq:ER_mrc} can reduce to 
    \begin{align}
        \label{eq:ER_mrc2}
        \bar \cR^{\rm mrc}
        & \! \approx \!N_u \log_2\!\Bigg(\!1+ \frac{ \rho\alpha_b N_{\rm RF}(1+1/L)}{\kappa + \rho (N_u-1) + 2\rho (1-\alpha_b)}\!\Bigg).
        % &  \approx N_u \log_2\Bigg(1+ \frac{ \rho\alpha N_{\rm RF}(1+L)}{cL + \rho \alpha L  (N_u -1) + \rho (1-\alpha)(1+LN_u)}\Bigg).
    \end{align}
    The ergodic sum rate in \eqref{eq:ER_mrc2} achieves the optimal scaling law $\sim N_u \log N_{\rm RF}$  with respect to $N_{\rm RF}$  as in \eqref{eq:C_opt}.  
\end{remark}

Remark \ref{rm:scaling_law_MRC} shows that the optimal scaling law can be achieved by the proposed two-stage analog combining algorithm even with the practical baseband combiner.
%, such as the MRC combiner.
This result verifies that the two-stage analog combining architecture is effective to enhance the achievable rate in mmWave hybrid MIMO systems with low-resolution ADCs. 
%Note that the one-stage analog combiner without $\bW_{\rm RF_2}$ cannot achieve the optimal scaling law and have the performance limit even with any sophisticated combiner used in the baseband due to the quantization error. 
To specify the effect of employing the second analog combiner $\bW_{\rm RF_2}$, we also derive the ergodic rate \eqref{eq:E_rk} without using $\bW_{\rm RF_2}$.
% in Corollary \ref{cor:ergodic_rate_one_mrc}. 
\begin{corollary}\label{cor:ergodic_rate_one_mrc}
     For the considered mmWave channel with low-resolution ADCs, the MRC ergodic rate of the ARV-TSAC without the second analog combiner is approximated as
    \begin{align}
    	\nonumber
        &\bar \cR^{\rm mrc}_{\rm one}\\
        % & \approx N_u\log_2\Bigg(1+ \frac{ \rho\alpha N_r N_{\rm RF}(1+L)}{LN_{\rm RF} + \rho \alpha L N_r (N_u -1) + \rho (1-\alpha)LN_r(2N_{\rm RF} + N_u -1)}\Bigg) \\
        \label{eq:ER_mrc_one}
        &\! \approx\! N_u\!\log_2\!\Bigg(\!1\!+\! \frac{ \rho\alpha_b N_r N_{\rm RF}(1+1/L)}{ N_{\rm RF} \!+\! \rho   N_r (N_u \!-\!1) \!+\! 2\rho (1\!-\!\alpha_b)N_r N_{\rm RF}/L }\!\Bigg).
    \end{align}
    \normalsize
\end{corollary}
\begin{proof}
   See Appendix \ref{appx:ergodic_rate_one_mrc}.
\end{proof}
Unlike the quantization noise term $2\rho(1-\alpha_b) N_r$ in \eqref{eq:ER_mrc}, that $2\rho(1-\alpha_b) N_r N_{\rm RF}/L$ in \eqref{eq:ER_mrc_one} includes $N_{\rm RF}/L$, which prevents the optimal scaling of the ergodic sum rate as in \eqref{eq:C_opt} with respect to $N_{\rm RF}$ for fixed $L$.  
%As expected, we have the following remark. 
\begin{remark}\label{rm:no_scaling_MRC}
    Let $\kappa = N_{\rm RF} / N_r $ where $\kappa \in(0,1)$ is a constant value.
    Then, \eqref{eq:ER_mrc_one} can reduce to 
    \begin{align}
        \label{eq:ER_mrc_one2}
       \bar \cR^{\rm mrc}_{\rm one}\!\! \approx\!\! N_u\!\log_2\!\Bigg(\!\!1\!+\! \frac{ \rho\alpha_b  N_{\rm RF}(1+1/L)}{\kappa \!+\! \rho   (N_u \!-\!1)\! +\! 2 \rho(1\!-\!\alpha_b)N_{\rm RF}/L }\!\!\Bigg).
    \end{align}
     Note that unlike the ergodic rate of the two-stage analog combining $\bar \cR^{\rm mrc}$ in \eqref{eq:ER_mrc2}, that of the one-stage analog combining $\bar \cR^{\rm mrc}_{\rm one}$ in \eqref{eq:ER_mrc_one2} cannot achieve the optimal scaling law with respect to the number of RF chains $N_{\rm RF}$.
\end{remark}

\section{Simulation Results}
\label{sec:sim}
%%%%%%%%%%%%%%%%%%%%%%%%%%%%%%%%%%%%%%%%%%%%%%%%%%%%%%%

In this section, we evaluate the performance of the proposed two-stage analog combing algorithm in the MI and ergodic sum rate. 
In the simulations, we set the codebook size to be $|\cV|=N_r$, which guarantees $\bW_{\rm RF}^H\bW_{\rm RF} = \bI_{N_{\rm RF}}$.  Consequently, analog combiners used in the simulations are semi-unitary.
To provide a reference performance of a conventional one-stage analog combining approach, we simulate a greedy-based MI maximization method which solves the following problem for the given ARV codebook in a greedy way:
\begin{align*}
	\cP2:& ~~\qquad\qquad \bW_{\rm RF}^{\rm opt,c}=\argmax_{\bW_{\rm RF}} ~\cC(\bW_{\rm RF})\\
	&\text{s.t.} \quad\bW_{\rm RF}^H\bW_{\rm RF}=\bI,~~ |[\bW_{\rm RF}]_{i,j}| = \frac{1}{\sqrt{N_r}}, \forall i,j.
\end{align*} 
At each iteration, the greedy method searches for a single ARV from the codebook $\cV$ which maximizes the MI with the previously selected ARVs and thus can nearly provide the optimal MI performance of the one-stage analog combining for the given codebook. 

In the simulations, we evaluate the following cases:
%Including the greedy algorithm, we evaluate the following cases:
\begin{enumerate}
    \item ARV-TSAC: proposed two-stage analog combining.
    % case with the proposed ARV-TSAC method.
    \item ARV: one-stage analog combining with $\bW_{\rm RF} = \bW_{\rm RF_1}$ selected from the ARV-TSAC.
    % case where the analog combiner $\bW_{\rm RF} = \bW_{\rm RF_1}$ is obtained from  the ARV-TSAC algorithm.
    % Algorithm \ref{algo:ARV}.
    \item SVD+DFT: two-stage analog combining with $\bW_{\rm RF_1} = \bU_{1:N_{\rm RF}}$ and $\bW_{\rm RF_2} = \bW_{\rm DFT}$ based on Theorem \ref{thm:optimality_two_stage}.
    % the two-stage analog combining case where the first analog combiner is the matrix of the $N_{\rm RF}$ left singular vectors $ \bW_{\rm RF_1} = \bU_{1:N_{\rm RF}}$ without any constraint and the second analog combiner is the normalized DFT matrix $\bW_{\rm RF_2} = \bW_{\rm DFT}$.
    \item SVD: one-stage analog combining $\bW_{\rm RF} = \bU_{1:N_{\rm RF}}$.
    % case where the analog combiner is $\bW_{\rm RF} = \bU_{1:N_{\rm RF}}$ without any constraint.
    \item Greedy-MI: one-stage analog combining with greedy-based MI maximization.
    % \item FD: the fully-digital low-resolution ADC case with $N_{\rm RF}$ antennas and RF chains.
\end{enumerate}
% The greedy method consists of the following steps:
% \begin{enumerate}
%     \item Initialize an analog combiner $\bW_{\rm RF}$ as an empty matrix and a spatial angle codebook $\cV$.
%     \item Search for a spatial angle in $\cV$ as $\vartheta^\star=\argmax_{\vartheta \in \mathcal{V}} \cC([\bW_{\rm RF}, \ba(\vartheta)])$ where $\cC$ is in \eqref{eq:MI}.
%     \item Update $\bW_{\rm RF} = [\bW_{\rm RF}, \ba(\vartheta^\star)]$ and $\cV = \cV\setminus\{\vartheta^\star\}$, and repeat step 2 - 3 until $N_{\rm RF}$ iterations.
% \end{enumerate}
% the algorithm selects a single ARV from the codebook $\cV$ that maximizes the MI  in the set of analog combiners composed of an ARV and the selected ARVs in the previous iteration. good !!
% For the ARV-TSAC and greedy algorithms, we set the codebook size to be $|\cV|=N_r$ to guarantee $\bW_{\rm RF}^H\bW_{\rm RF} = \bI_{N_{\rm RF}}$.
% Consequently, all analog combiners used in the simulations are semi-unitary.
The SVD+DFT and SVD cases are infeasible in practice due to violating the constant modulus constraint, and SVD+DFT provides a tight upper bound on MI for a homogeneous singular value case from Theorem \ref{thm:equal_eigenvalue_optimal}.
Here, we adopt $L_k = {\rm max}\{1, {\rm Poisson}(\lambda_L)\}$ \cite{akdeniz2014millimeter} unless mentioned otherwise, where $\lambda_L$ is considered as the average number of channel paths.

% FIGURE %%%%%%%%%%%%%%%%%%%%%%%%% 
\begin{figure}[!t]\centering
	\includegraphics[scale = 0.34]{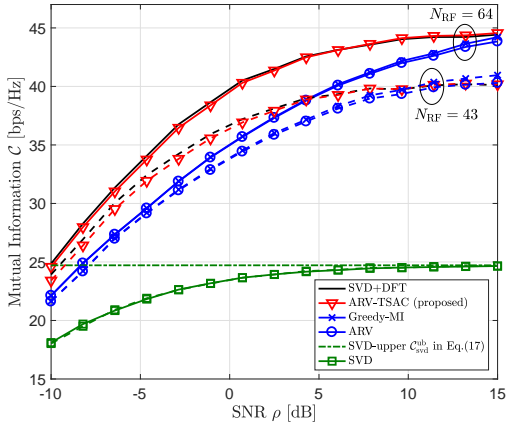}
	\caption{The MI simulation results for $N_r = 128$ receive antennas,  $N_u = 8$ users, $\lambda_L = 3$ average channel paths, $b=2$ quantization bits, and $N_{\rm RF} \in \{43, 64\}$ RF chains that are  $\lceil N_r/3 \rceil$ and $\lceil N_r/2\rceil$, respectively. }
	\label{fig:MI_SNR}
	\vspace{-1em}
\end{figure}
%%%%%%%%%%%%%%%%%%%%%%%%%%%%%%%%%%

% FIGURE %%%%%%%%%%%%%%%%%%%%%%%%%
\begin{figure}[t]
\centering
$\begin{array}{c c}
{\resizebox{0.9\columnwidth}{!}
{\includegraphics{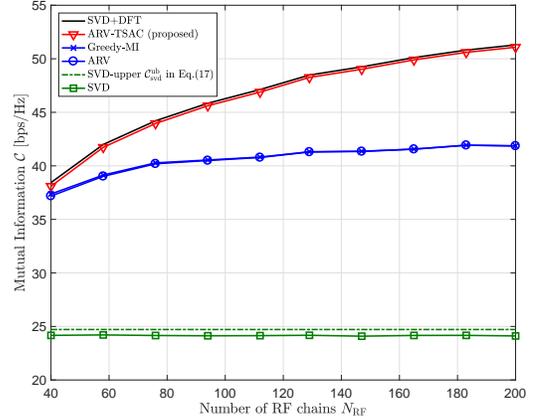}}
}\\ \mbox{\small(a) $N_r = 256$}\\
{\resizebox{0.9\columnwidth}{!}
{\includegraphics{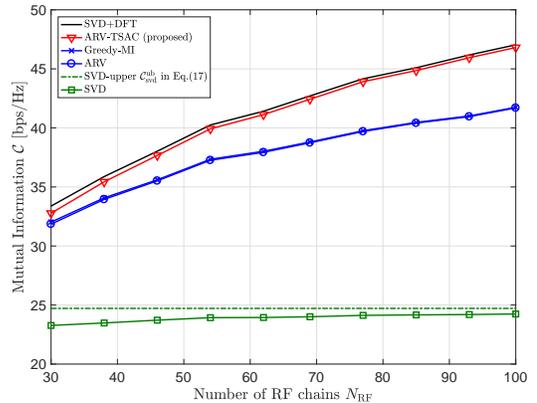}}
}\\\mbox{\small(b) $\kappa = 1/3$}
\end{array}$
\caption{The MI simulation results with $N_u = 8$ users, $\lambda_L = 4$ average channel paths, $b=2$ quantization bits, and $\rho = 0$ dB SNR for (a) $N_r = 256$ receive antennas and (b) $\kappa = N_{\rm RF}/N_r = 1/3$.}
\label{fig:MI_NRF}
\vspace{ -1 em}
\end{figure}
%%%%%%%%%%%%%%%%%%%%%%%%%%%%%%%%%%

% FIGURE %%%%%%%%%%%%%%%%%
\begin{figure*}[!t]\centering
\centering
$\begin{array}{c c c}
{\resizebox{0.65\columnwidth}{!}
{\includegraphics{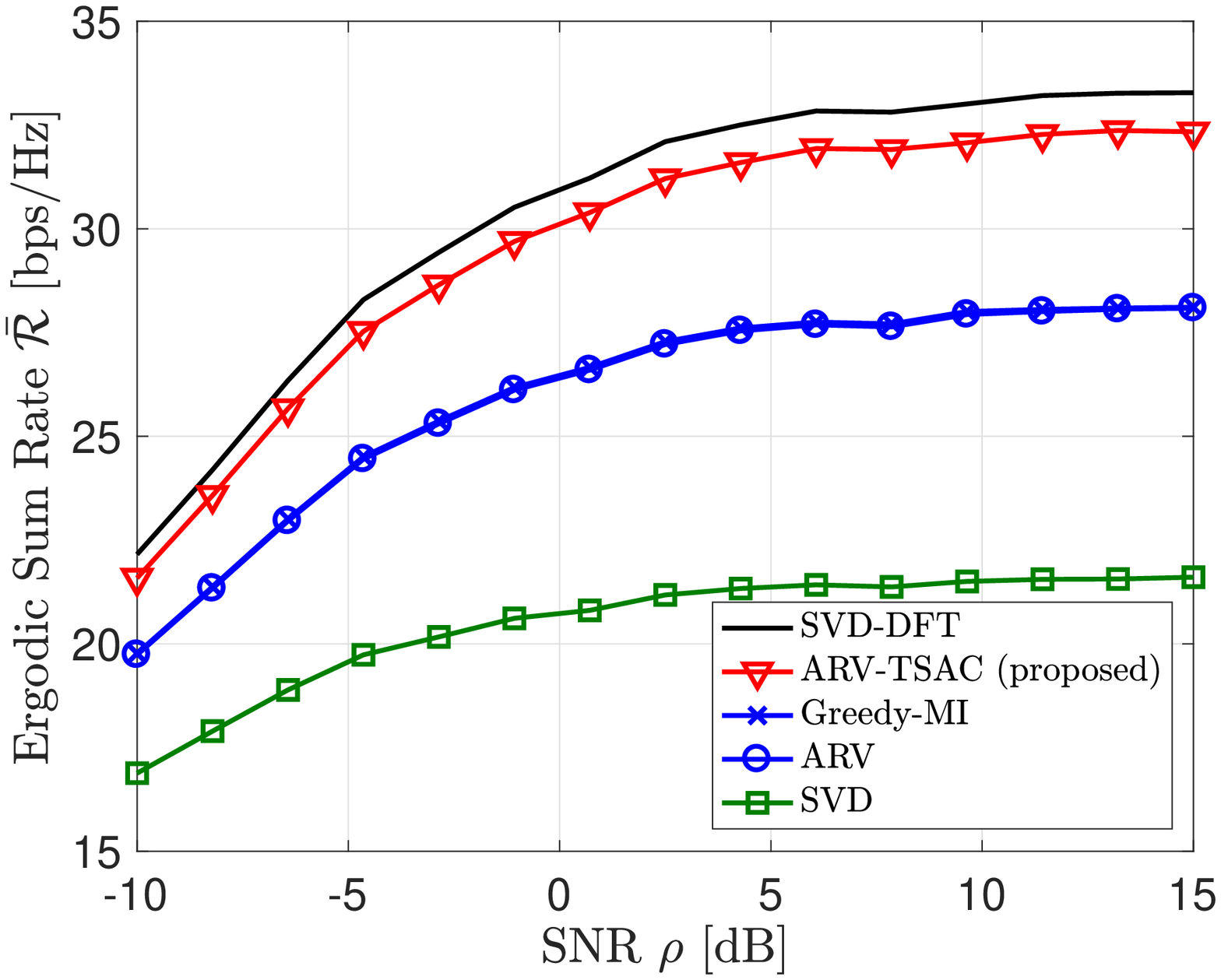}}
}&
{\resizebox{0.65\columnwidth}{!}
{\includegraphics{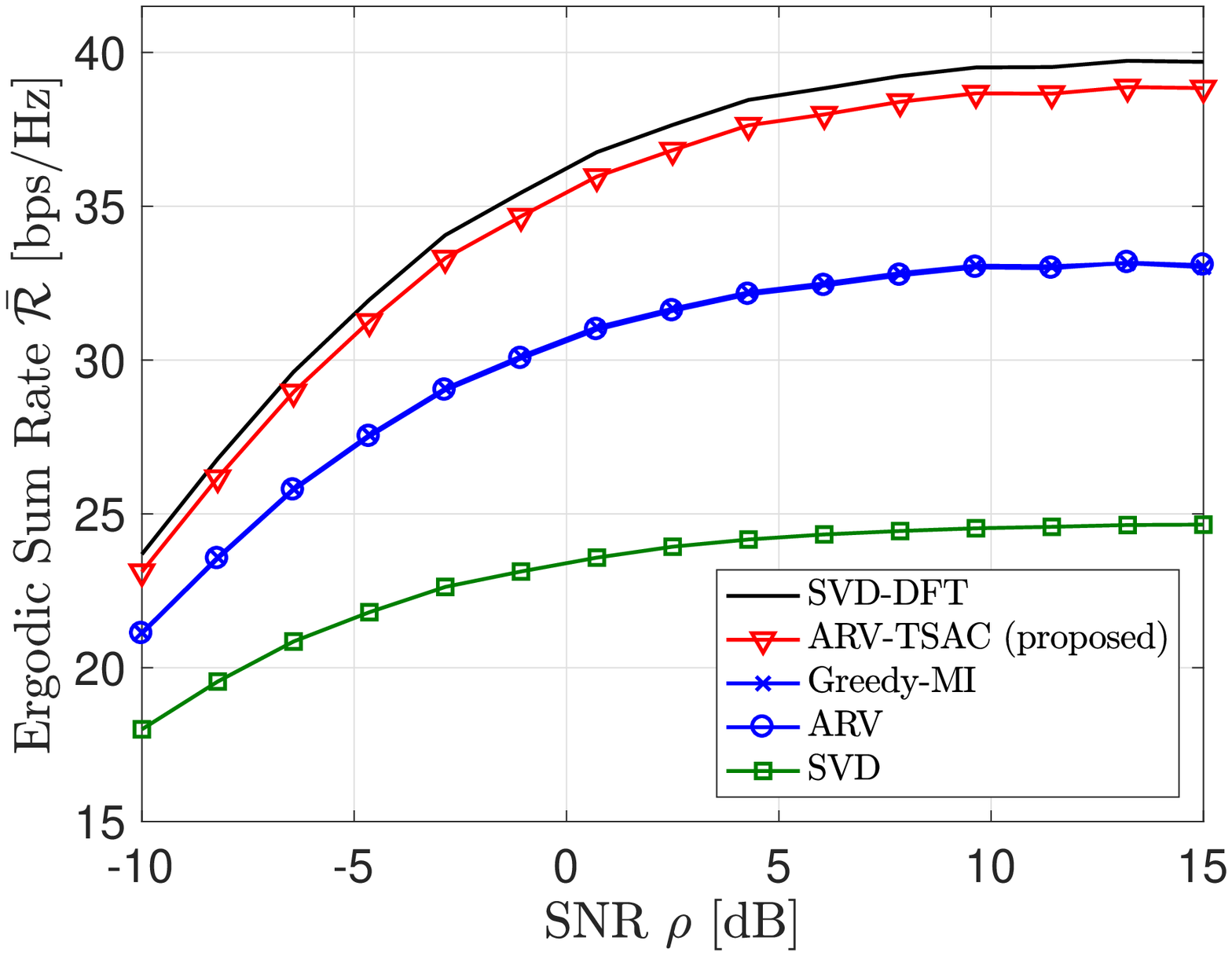}}
}&
{\resizebox{0.65\columnwidth}{!}
{\includegraphics{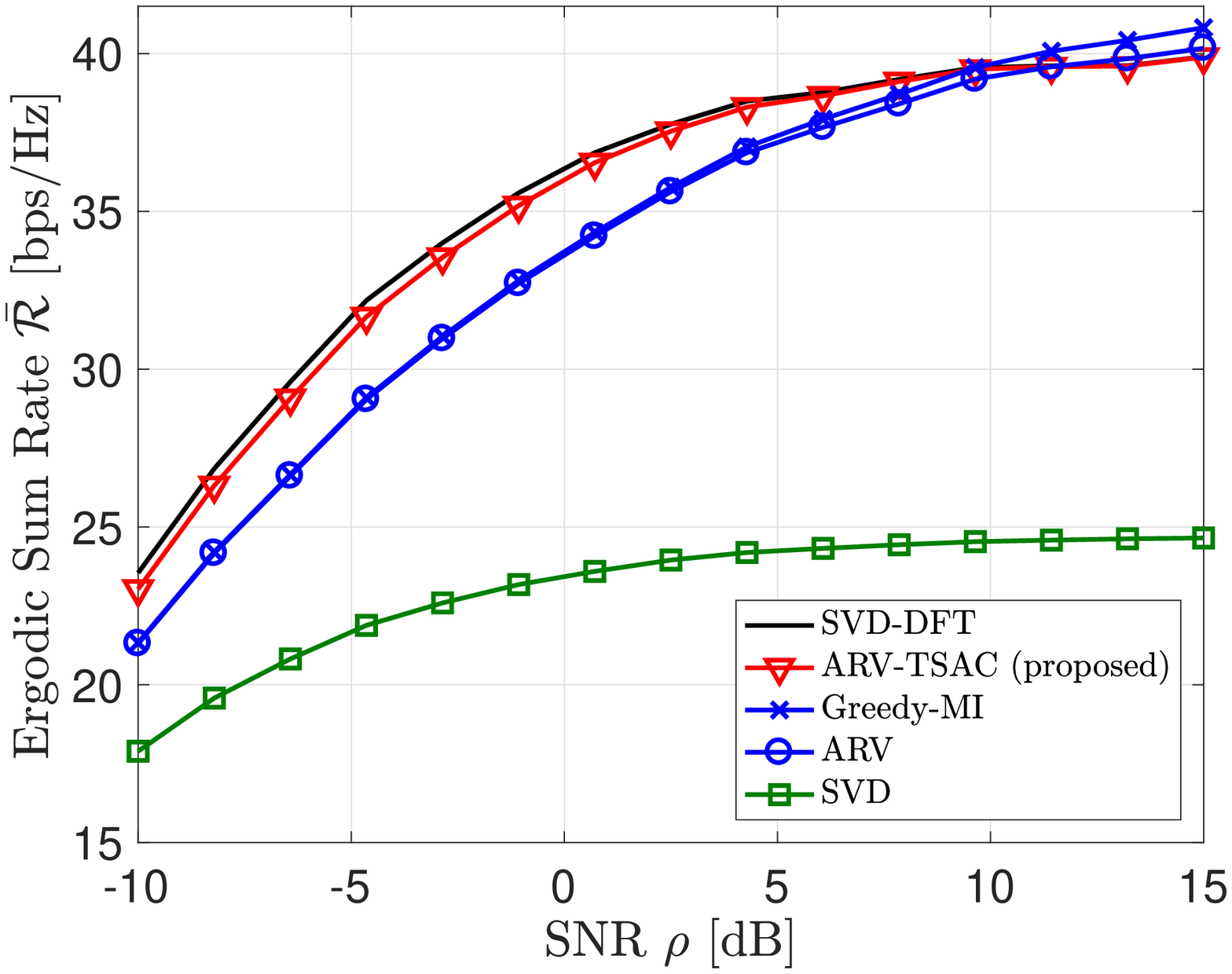}}
}\\
\mbox{\small (a) MRC} & \mbox{\small (b) ZF}  & \mbox{\small (c) MMSE}
\end{array}$
         \caption{Simulation results of the ergodic sum rate with $N_r = 128$ receive antennas, $N_{\rm RF} = 43$ RF chains, $N_u =8$ users, $\lambda_L = 3$ average channel paths, and $b=2$ quantization bits for (a) maximum ratio combining (MRC), (b) zero-forcing (ZF), and (c) minimum mean squared error (MMSE) digital combiners.} 
	\label{fig:Rate_WBB}
% 	\vspace{-1 em}
\end{figure*}
%%%%%%%%%%%%%%%%%%%%%%%%%`  

% \vspace{-1 em}
%%%%%%%%%%%%%%%%%%
\subsection{Mutual Information}
%%%%%%%%%%%%%%%%%

Fig. \ref{fig:MI_SNR} shows the MI simulation results for $N_r = 128$, $N_{\rm RF} \in \{43, 64\}$, $N_u = 8$, $\lambda_L = 3$, and $b=2$ with respect to the SNR $\rho$.
The proposed ARV-TSAC algorithm achieves a similar MI as does the SVD+DFT case, and they show the best MI over the most SNR values.
The Greedy-MI and ARV cases provide similar MI to each other but show the MI gap from the ARV-TSAC. 
The gap decreases as $\rho$ increases in the high SNR regime, and the Greedy-MI and ARV cases with $N_{\rm RF} = 43$ show the higher MI than SVD+DFT and ARV-TSAC in the very high SNR regime.
Such phenomenon occurs as the channel environment does not guarantee the optimality condition for the two-stage analog combining solution in Theorem \ref{thm:equal_eigenvalue_optimal}.
As more RF chains are used, however, the MI gap between ARV-TSAC/SVD+DFT and Greedy-MI/ARV becomes larger and the performance reversal would happen in even the higher SNR regime.
This is because the proposed two-stage analog combining can exploit more RF chains to further reduce quantization errors.
The SVD case results in the worst MI performance and it converges to the theoretic upper bound $\cC_{\rm svd}^{\rm ub}$ due to the quantization error.
% Later in this section, we show that the SVD case achieves the best performance without quantzation error.

Fig. \ref{fig:MI_NRF} shows the MI simulation results with $N_u = 8$, $\lambda_L = 4$, $b=2$, and $\rho = 0$ dB in terms of  $N_{\rm RF}$.
% for the different number of RF chains $N_{\rm RF}$.
% Note that the proposed ARV-TSAC algorithm achieves the similar MI to that of the SVD+DFT case for the considered system. 
In Fig. \ref{fig:MI_NRF}(a), $N_r$ is fixed to be $N_r = 256$. 
The two-stage combining cases, i.e., SVD+DFT and ARV-TSAC, show  that the MI increases logarithmically with $N_{\rm RF}$, and this corresponds to the scaling law derived in Theorem \ref{thm:optimality_two_stage}. 
The one-stage combining cases such as the Greedy-MI, ARV, and SVD cases, however, show a marginal increase of the MI as $N_{\rm RF}$ increases. 
In Fig. \ref{fig:MI_NRF}(b), the ratio between $N_r$ and $N_{\rm RF}$ is fixed to be $\kappa = 1/3$.
Here, the Greedy-MI and ARV cases also increase more slowly compared to the SVD+DFT and ARV-TSAC cases.
This is because more channel gains can be collected as $N_r$ increases for all cases, but the two-stage combining can reduce more quantization error as $N_{\rm RF}$ increases.
Accordingly, the MI gap between the two-stage combining and one-stage combining cases increases as $N_{\rm RF}$ increases. 
% This is, again, because the two-stage combining can reduce more quantization errors by leveraging the RF chains.

% \vspace{ -1 em}
%%%%%%%%%%%%%%%%%%
\subsection{Ergodic Sum Rate}
%%%%%%%%%%%%%%%%%%

Now, we evaluate the ergodic rate for linear digital combiners $\bW_{\rm BB}$ such as MRC, zero-forcing (ZF), and MMSE.
%We also confirm the insight for the two-stage analog combining architecture with a linear receiver regarding the mitigation of quantization error. 
Let $\bH_{\rm eq} = \bW_{\rm RF}^H \bH$.
The MRC, ZF, and MMSE combiners are given as: $\bW_{\rm BB, mrc} = \bH_{\rm eq}, \bW_{\rm BB, zf} = \bH_{\rm eq}(\bH_{\rm eq}^H\bH_{\rm eq})^{-1}$, and $\bW_{\rm BB, mmse} =\bR_{\by_{\rm q}\by_{\rm q}}^{-1} \bR_{\by_{\rm q}\bx}$, where $\bR_{\by_{\rm q}\bx} = \alpha\rho \bH_{\rm eq}$ and $\bR_{\by_{\rm q}\by_{\rm q}}\! =\! \alpha^2\rho \bH_{\rm eq}\bH_{\rm eq}^H \! +\! \alpha^2\bW_{\rm RF}^H\bW_{\rm RF}\! +\! \bR_{\bq\bq}$. 
For the given analog and digital combiners $ (\bW_{\rm RF}, \bW_{\rm BB})$ with $\bW_{\rm RF}^H \bW_{\rm RF} = \bI_{N_{\rm RF}}$, the ergodic rate of user $k$ is expressed as
\begin{align} \nonumber
   & \bar r_k(\bW_{\rm RF},\!\bW_{\rm BB}) = \bbE\Big[\log_2\left(1+{\alpha_b^2\rho|\bw_{{\rm BB},k}^H\bh_{{\rm eq},k}|^2}/{\eta_{{\rm BB},k}}\right)\Big]
    %  \bar r_k(\bW_{\rm RF},\!\bW_{\rm BB})\! =\! \bbE\!\left[\log_2\!\left(\!1\!+\!\frac{\alpha^2\rho|\bw_{{\rm BB},k}^H\bh_{{\rm eq},k}|^2}{\alpha^2\rho\sum_{u\neq k}^{N_u}|\bw_{{\rm BB},k}^H\bh_{{\rm eq},u}|^2 + \alpha^2\|\bW_{\rm RF}\bw_{{\rm BB},k}\|^2 + \bw_{{\rm BB},k}^H\bR_{\bq\bq}\bw_{{\rm BB},k}}\!\right)\!\right]
\end{align}
where $\eta_{{\rm BB},k} = \alpha_b^2\rho\sum_{u\neq k}^{N_u}|\bw_{{\rm BB},k}^H\bh_{{\rm eq},u}|^2 + \alpha_b^2\|\bw_{{\rm BB},k}\|^2 + \bw_{{\rm BB},k}^H\bR_{\bq\bq}\bw_{{\rm BB},k}$.

Fig. \ref{fig:Rate_WBB} illustrates the ergodic sum rates with $N_r = 128$, $N_{\rm RF} = 43$, $N_u =8$, $\lambda_L = 3$, and $b=2$  versus the SNR $\rho$ for different digital combiners: (a) MRC, (b) ZF, and (c) MMSE.
Similarly to the MI results, ARV-TSAC shows the comparable ergodic rate to that of SVD+DFT and outperforms the one-stage combining such as the Greedy-MI and ARV cases in most cases.
We note that the SVD case also shows the worst sum rate performance in the considered systems.
The gaps between the two-stage combining cases and one-stage combining cases for the MRC and ZF combiners are much larger than the gap for the MMSE combiner. 
In addition, SVD+DFT and ARV-TSAC with the ZF combiner achieve the ergodic rates comparable to the MMSE combiner, while the Greedy-MI and ARV cases with the ZF combiner show much lower ergodic sum rates than that with the MMSE combiner.
Since the MRC and ZF combiners ignore the AWGN and quantization noise whereas the MMSE combiner does not, using the MMSE combiner improves the ergodic rate of the one-stage analog combining cases.
The two-stage analog combining cases, however, already reduced the quantization noise by using the second analog combiner, and thus, they provide the MMSE-like ergodic rate performance with the ZF combiner.
Therefore, the proposed two-stage analog combining with the ARV-TSAC algorithm can achieve significant rate improvement with the MRC or ZF combiners compared to the one-stage analog combining approach.

% % FIGURE %%%%%%%%%%%%%%%%%%%%%%%%%
% \begin{figure}[t]
% \centering
% $\begin{array}{c c c}
% {\resizebox{0.3\columnwidth}{!}
% {\includegraphics{Analysis_Nr_kappa.png}}
% }&
% {\resizebox{0.3\columnwidth}{!}
% {\includegraphics{Analysis_Nr_fix.png}}
% }&
% {\resizebox{0.3\columnwidth}{!}
% {\includegraphics{Analysis_Nr_fix.png}}
% }\\
% \mbox{(a)} & \mbox{(b)} & \mbox{(b)}
% \end{array}$
% \caption{Simulation results for the ergodic sum rate of the MRC digital combiner with $N_u =8$ users, $\lambda_L = 4$ average channel paths, and $b=2$ quantization bits for (a) $N_r = 256$ receive antennas and (b) $N_{\rm RF}/N_r \approx 1/3$.}
% \label{fig:Analysis}
% \end{figure}
% %%%%%%%%%%%%%%%%%%%%%%%%%%%%%%%%%%
 
Fig. \ref{fig:Rate_NRF} provides the simulation results of the ergodic rate with the MRC digital combiner for $N_u =8$, $\lambda_L = 3$, and $\rho = 0$ dB in terms of the number of (a) RF chains $N_{\rm RF}$ and (b) quantization bits $b$.
In Fig. \ref{fig:Rate_NRF}(a), we consider $b = 2$ and $\kappa = N_{\rm RF}/N_r = 1/3$.
The ergodic rates of SVD+DFT and ARV-TSAC are similar and both increase logarithmically with $N_{\rm RF}$, whereas the ergodic rates of the Greedy-MI and ARV cases increase more slowly.
Such scaling results correspond to Remark \ref{rm:scaling_law_MRC} and \ref{rm:no_scaling_MRC}. 
% Although the ergodic rates of Greedy-MI and ARV as well as that of SVD+DFT and ARV-TSAC increase logarithmically with $N_r$, SVD+DFT and ARV-TSAC provide the higher ergodic rates and the rate gap becomes larger as $N_r$ increase.
As $N_r$ increases with a fixed $\kappa$, SVD+DFT and ARV-TSAC effectively reduce the more quantization error while obtaining larger channel gains, but the Greedy-MI and ARV cases only obtain larger channel gains without mitigating the quantization error.
% The SVD case also shows a similar trend as the Greedy-MI and ARV cases with a large sum rate gap.
In Fig. \ref{fig:Rate_NRF}(b), we consider $N_r = 128$ and $ N_{\rm RF} = 43$. 
We note that in the low-resolution ADC regime, the ARV-TSAC algorithm achieves the ergodic rate comparable to that of SVD+DFT and shows a noticeable improvement compared to the Greedy-MI, ARV, and SVD cases. 
As $b$ increases, the ergodic rates of the ARV-TSAC, Greedy-MI, and ARV algorithms converge to each other with a small gap from the SVD+DFT case.
The ergodic rate of the SVD case, however, converges to that of SVD+DFT without any gap because the SVD combining is optimal in maximizing the MI of infinite-resolution ADC systems. 
The simulation results validate the effectiveness of the proposed two-stage combining in low-resolution ADC systems.
% Accordingly, the results in Fig. \ref{fig:Rate_NRF} demonstrate that the two-stage combining approach effectively manages the quantization error from the low-resolution ADCs with practical digital combiners.

% FIGURE %%%%%%%%%%%%%%%%%%%%%%%%% 
\begin{figure}[t]
\centering
$\begin{array}{c}
{\resizebox{0.93\columnwidth}{!}
{\includegraphics{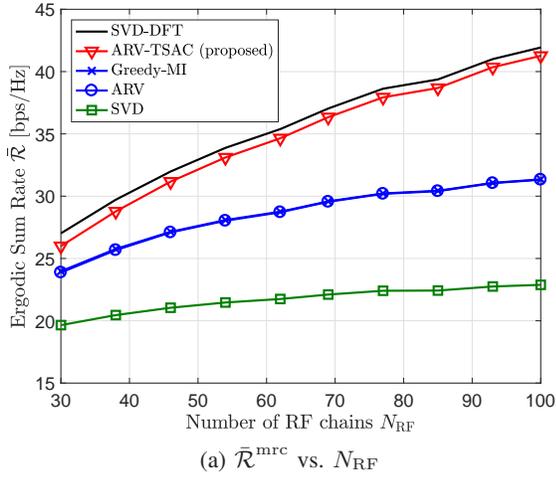}}
}\\ \mbox{\small (a) $\bar \cR^{\rm mrc}$ vs. $N_{\rm RF}$} \\
{\resizebox{0.93\columnwidth}{!}
{\includegraphics{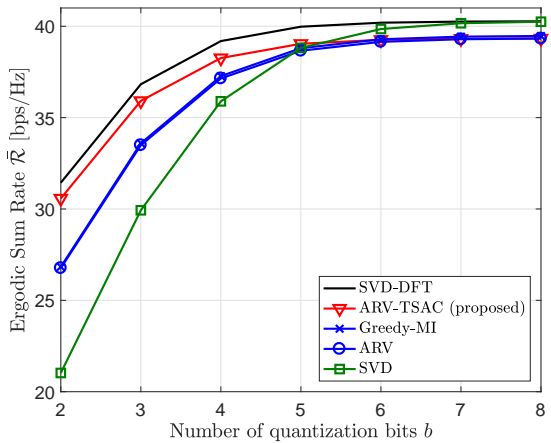}}
}\\ \mbox{\small (b)  $\bar \cR^{\rm mrc}$ vs. $b$} 
\end{array}$
\caption{Simulation results of the ergodic sum rate of the MRC combiner $\bar \cR^{\rm mrc}$ with $N_u =8$ users, $\lambda_L = 3$ average channel paths, and $\rho = 0$ dB SNR for (a) $b=2$ quantization bits and $\kappa = N_{\rm RF}/N_r  = 1/3$ and (b) $N_r =128$ receive antennas and $N_{\rm RF} = 43$ RF chains.}
\label{fig:Rate_NRF}
% \vspace{-1 em}
\end{figure}
%%%%%%%%%%%%%%%%%%%%%%%%%%%%%%%%%%

Finally, we validate the derived ergodic rates in Theorem \ref{thm:ergodic_rate_mrc} and Corollary \ref{cor:ergodic_rate_one_mrc}. 
We consider $N_r = 128$ receive antennas, $N_{\rm RF} = 43$ RF chains, $N_u = 8$ users each with $L = 8$ channel paths for the virtual channels, and $b = 2$ quantization bits. 
In Fig. \ref{fig:Analysis}, the theoretical ergodic rates tightly align with the simulation results in the medium to high SNR regime, and show similar trend as the simulation results do.
Thus, the derived ergodic rates can characterize the ergodic rate performance of the proposed algorithm for the two-stage analog combining system in terms of the system parameters including quantization resolution.

Overall, the two-stage analog combining structure with the ARV-TSAC algorithm almost achieves the performance of SVD+DFT that is a near optimal solution for the unconstrained problem $\cP1$, while the greedy-MI and ARV algorithms provide a near optimal solution only for the constrained problem $\cP2$. 
Since $\cP1$ has a larger feasible set than $\cP2$ to find an optimal solution for the same objective function, this leads to $\cC(\bW_{\rm RF}^{\rm opt}) \geq \cC(\bW_{\rm RF}^{\rm opt,c})$.
In this regard, the ARV-TSAC algorithm achieves the higher performance than that of the Greedy-MI and ARV algorithms in most cases.
This shows that the proposed two-stage analog combining architecture with the ARV-TSAC is a practical solution suitable for the mmWave hybrid MIMO systems with low-resolution ADCs.

% \subsection{Energy Efficiency}

% {\color{blue}
% \begin{enumerate}
%     \item Practical BB combiners (Real): 
%     \begin{enumerate}
%         \item ARV + MRC/ZF/MMSE
%         \item ARV+DFT + MRC/ZF/MMSE 
%         \item FD + MRC/ZF/MMSE
%     \end{enumerate}
% \end{enumerate}
% }

%%%%%%%%%%%%%%%%%%%%%%%%%%%%%%%%%%%%%%%%%%%%%%%%%%%%%%%
\section{Conclusion}
\label{sec:con}
%%%%%%%%%%%%%%%%%%%%%%%%%%%%%%%%%%%%%%%%%%%%%%%%%%%%%%%

In this paper, we derived a near optimal analog combining solution for an unconstrained MI maximization problem in hybrid MIMO systems with low-resolution ADCs. 
We showed optimalities of the solution in the scaling law and in maximizing the mutual information for a homogeneous channel singular value case.
To implement the derived solution, we proposed a two-stage analog combining architecture that decouples the channel gain aggregation and spreading functions in the solution into two cascaded analog combiners.
Accordingly, the proposed two-stage analog combining also provides a near optimal solution for the unconstrained problem whereas conventional hybrid algorithms offer a near optimal solution only for the constrained problem.
In addition, we derived a closed-form approximation to the ergodic rate, which reveals that our two-stage analog combiner achieves the optimal scaling law with a practical digital combiner.
% We further analyzed the performance of the proposed algorithm in ergodic rate by deriving a closed-form approximation.
% The derived ergodic rate reveals that the two-stage analog combiner from the algorithm also achieves the optimal scaling law with a practical digital combiner.
Simulation results validated the key insights obtained in this paper and the derived ergodic rate, and also demonstrated that the proposed two-stage analog combining algorithm outperforms conventional algorithms. 
% {\color{blue}
% The proposed two-stage analog combining structure provides a near optimal solution for an unconstrained mutual information maximization problem which only imposes a unitary condition on an analog combiner without a constant modulus constraint.
% The conventional analog combining structure, however, only offers a near optimal solution for a constrained MI maximization problem. 
% Therefore, the proposed ARV two-stage analog combining algorithm outperforms conventional algorithms.
% As shown in this paper, the two-stage analog combining structure with the ARV-TSAC nearly provides a near optimal solution for $\cP1$ for a given ARV codebook, while the greedy-MI and ARV algorithms offer a near optimal solution only for $\cP3$. 
% Since $\cP1$ has a larger feasible set than $\cP3$ to find an optimal solution, the optimal solutions lead to $\cC(\bW_{\rm RF}^{\rm opt}) \geq \cC(\bW_{\rm RF}^{\rm opt,c})$.
% In this regards, the ARV-TSAC algorithm achieves the higher performance than that of the Greedy-MI and ARV cases in most cases.
% This shows that the proposed two-stage analog combining structure is more suitable for the hybrid MIMO systems with low-resolution ADCs.
% }

% FIGURE %%%%%%%%%%%%%%%%%%%%%
\begin{figure}[!t]\centering
	\includegraphics[scale = 0.42]{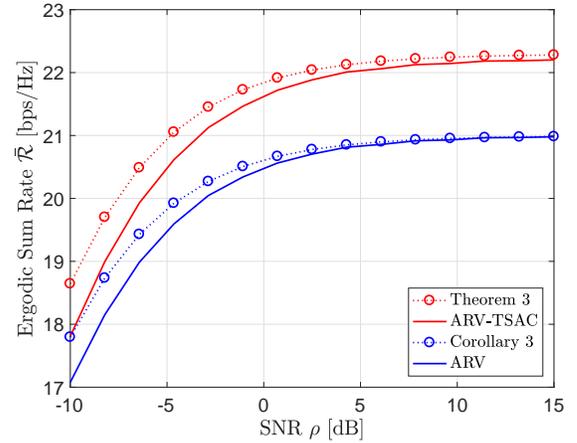}
	\caption{Comparison of the ergodic rate for the theoretical and simulation results with $N_r = 128$ receive antennas, $N_{\rm RF} = 43$ RF chains, $N_u = 8$ users each with $L = 8$ channel paths for the virtual channels.} 
	\label{fig:Analysis}
% 	\vspace{-1 em}
\end{figure}
%%%%%%%%%%%%%%%%%%%%%%%%%%%%%

\begin{appendices}

%%%%%%%%%%%%%%%%%%%%%%%%%%%%
\section{Proof of Corollary \ref{cor:AoA_DFT}}
\label{appx:SVDtoARV}
%%%%%%%%%%%%%%%%%%%%%%%%%%%%

  Let $\bH$ be decomposed into $\bH=\bA_{\rm AoA}\bH_{\rm V}$, where $\bH_{\rm V}={\rm blkdiag}\{\tilde{\bg}_1,\cdots,\tilde{\bg}_{N_u}\}$ and $\tilde{\bg}_k=\sqrt{\frac{N_r}{L_k}}[g_{1,k},\cdots,g_{L_k,k}]^T$. Then, it can be shown \cite{ngo2014aspects} that   as  $N_r \to \infty$,
    \begin{align}\label{eq:AoA_convergence}
        \bW_{\rm AoA}^H\!\bW_{\rm AoA} \!\to\! \bI_{N_{\rm RF}},~~\frac{1}{\sqrt{N_r}}\bW_{\rm AoA}^H\!\bH \!\to \!\frac{1}{\sqrt{N_r}}\!
        \left[\begin{matrix}
        \bH_{\rm V} \\
        {\bf 0}
        \end{matrix}
        \!\right].
    \end{align}
    Let $\tilde{\bH}_{\rm V}=[\bH_{\rm V}^T, {\bf 0}^T]^T$ and $ \bC_{\rm AoA}= \bW_{\rm RF_2}^{\star H}\tilde{\bH}_{\rm V}\tilde{\bH}_{\rm V}^H\bW_{\rm RF_2}^{\star}$. 
    Using \eqref{eq:AoA_convergence}, we show $\cC(\bW_{\rm RF})$ in \eqref{eq:c_wrf_wlambdaw} with $\bW_{\rm RF}=\tilde\bW_{\rm RF}^\star$ converges as $N_r \to \infty$ to 
    \begin{align} 
    	 \label{eq:W_AOA_convergence}
        &\left(\cC(\tilde\bW_{\rm RF}^\star)\! -\! \log_2\!\Big|\bI\! +\! \frac{\alpha_b}{\beta_b} {\rm diag}^{\!-1}\!\left\{\bC_{\rm AoA}\!+\!\tfrac{1}{\beta_b\rho}\bI\right\}\!\bC_{\rm AoA} \Big|\right)       \to 0.
    \end{align}
%   {\small \begin{align} 
%       \cC(\bW_{\rm RF}') -\log_2\left|\bI_{N_{\rm RF}} + \frac{\alpha}{1-\alpha} {\rm diag}^{-1}\left\{\bW_{\rm RF_2}^{\star H}\tilde{\bH}_{\rm V}\tilde{\bH}_{\rm V}^H\bW_{\rm RF_2}^{\star}+\tfrac{1}{(1-\alpha)\rho}\bI_{N_{\rm RF}}\right\}\bW_{\rm RF}^{'H}\bH\bH^H\bW_{\rm RF}' \right| \to 0. \label{eq:W_AOA_convergence}
%   \end{align}}
  Note that each diagonal of $\bW_{\rm RF_2}^{\star H}\tilde{\bH}_{\rm V}\tilde{\bH}_{\rm V}^H\bW_{\rm RF_2}^{\star}$ cannot exceed $\frac{1}{\kappa}\sum_{k=1}^{N_u}\frac{1}{L_k}(\sum_{\ell=1}^{L_k}|g_{\ell,k}|)^2=c_1 < \infty$. Let $\cC_{\rm \infty}(\tilde\bW_{\rm RF}^\star)$ denote the second term in \eqref{eq:W_AOA_convergence}. Then, $\cC_{\rm \infty}(\tilde\bW_{\rm RF}^\star)$ can be lower bounded as
  \begin{align}
      \cC_{\rm \infty}(\tilde\bW_{\rm RF}^\star) 
      &> \log_2\left|\bI_{N_{\rm RF}}+\frac{\alpha_b\rho}{c_1\beta_b\rho +1}\bW_{\rm RF_2}^{\star H}\tilde{\bH}_{\rm V}\tilde{\bH}_{\rm V}^H\bW_{\rm RF_2}^{\star}\right| \nonumber\\
      &\overset{(a)}{\sim} N_u \log_2N_{\rm RF}, ~\text{as}~ N_{\rm RF}\to \infty,\label{eq:lower_bound_of_conv_c}
  \end{align}
  where $(a)$ follows from the same reason of $(b)$ below \eqref{eq:achievable_rate_DFT_U}. This implies that $\cC(\tilde\bW_{\rm RF}^\star)$ follows the optimal scaling law. 
\qed
    
%%%%%%%%%%%%%%%%%%    
\section{Proof of Lemma \ref{lem:QN_auto}}
\label{appx:QN_auto}
%%%%%%%%%%%%%%%%%%

The auto quantization noise variance term in \eqref{eq:QN_auto} can be expressed as
    \begin{align}
        % \label{eq:auto_QN}
        \nonumber
        \bbE\big[&\Psi_k^{\rm auto}\big] = \bbE\!\left[ \sum_{i=1}^{N_{\rm RF}}\left|\bh_{{\rm b},k}^H \bw_i\right|^4 \right] = \left(\frac{N_r}{L}\right)^{\!2} \sum_{i=1}^{N_{\rm RF}}\bbE\!\left[\left|\bg_{k}^H \bw_i\right|^4 \right]\\
        \label{eq:auto_QN2}
        &= \left(\frac{N_r}{L}\right)^{\!\!2}\sum_{i=1}^{N_{\rm RF}}\!\bigg(\bbV\left[ \left|\bg_{k}^H \bw_i\right|^{2} \right]\!+\!\Big(\bbE\left[ \left|\bg_{k}^H \bw_i\right|^2 \right]\Big)^{\!2}\bigg)
        % & \stackrel{(a)}= \left(\frac{N_r}{L}\right)^2\sum_{i=1}^{N_{\rm RF}}\frac{1}{N_{\rm RF}^2}\bigg(\bbV\left[ \left|\bg_{k}^H \bw_i\right|^2 \right]+\Big(\bbE\left[ \left|\bg_{k}^H \bw_i\right|^2 \right]\Big)^2\bigg)
    \end{align}
    % where (a) follows from the constant modulus of $\bw_i$, i.e., $|w_{j,i}| = 1/\sqrt{N_{\rm RF}}$, $\forall i,j$, and from the IID complex Gaussian distribution for complex path gains $g_{\ell,k}$.
    where $\bw_i$ is the $i$th column of $\bW_{\rm DFT}$. 
    The expectation term  $\bbE[|\bg_{k}^H \bw_i|^2]$ in \eqref{eq:auto_QN2} is computed as 
    % $$\bbE[|\bg_{k}^H \bw_i|^2]=\bbE \big[\sum_{\ell=1}^{N_{\rm RF}}|g_{\ell,k}|^2 \big]/N_{\rm RF} = L/N_{\rm RF}$.
    \begin{align}
        \label{eq:Egw}
         \bbE\left[ \left|\bg_{k}^H \bw_i\right|^2 \right] =
    \frac{1}{N_{\rm RF}}\bbE\left[\sum_{\ell=1}^{N_{\rm RF}}|g_{\ell,k}|^2 \right] = \frac{L}{N_{\rm RF}}.
    \end{align}
    Now, let $\hat \bw_i = \sqrt{N_{\rm RF}}\bw_i$. 
    Then, we can compute the variance term $\bbV[|\bg_{k}^H \bw_i|^2]$ in \eqref{eq:auto_QN2} as 
    \begin{align}
        \nonumber
        &\bbV\!\left[ \left|\bg_{k}^H \bw_i\right|^2 \right]
        % & = \frac{1}{N_{\rm RF}^2} \bbV\left[\left|\sum_{\ell = 1}^{N_{\rm RF}}g_{\ell,k}^*\hat w_{\ell,i}\right|^2\right]\\
       \! =\! \frac{1}{N_{\rm RF}^2}\bbV\!\left[\sum_{\ell = 1}^{N_{\rm RF}}|g_{\ell,k}|^2 \!+\!\! \sum_{\ell_1 \neq \ell_2}^{N_{\rm RF}} \! g_{\ell_1,k}^*g_{\ell_2,k}\hat w_{\ell_1,i}^*\hat w_{\ell_2,i} \right] \\ \nonumber
        & \stackrel{(a)}= \frac{1}{N_{\rm RF}^2}\left( \bbV\!\left[\sum_{\ell = 1}^{N_{\rm RF}}|g_{\ell,k}|^2\right] + \bbV\!\left[\sum_{\ell_1 \neq \ell_2}^{N_{\rm RF}}g_{\ell_1,k}^*g_{\ell_2,k}\hat w_{\ell_1,i}^*\hat w_{\ell_2,i} \right]\right) \\
        \label{eq:Vgw}
        &  \stackrel{(b)}= \frac{1}{N_{\rm RF}^2}\left(\bbV\Big[\|\bg_k\|^2\Big] + \sum_{\ell_1 \neq \ell_2}^{N_{\rm RF}}\bbV\Big[g_{\ell_1,k}^*g_{\ell_2,k} \Big] \right)
    \end{align}
    where $(a)$ and $(b)$ hold as the associated terms are uncorrelated, which can be shown from straight forward mathematics, and $|\hat w_{\ell,i}| = 1$, $\forall \ell, i$.
    Since $\|\bg_k\|^2 \sim \chi^2_{2L}$, which is a chi-square distribution with $2L$ degrees of freedom,
    % , because $\|\bg_k\|^2$ is the sum of $L$ exponential random variables with unit mean and variance, 
    we have $\bbV[\|\bg_k\|^2] = L$, and $\bbV[g_{\ell_1,k}^*g_{\ell_2,k}]$ is computed as
    % Each diagonal entry of $\hat \bG_k$, $\bbV[g_{\ell,k}g_{\ell,k}^*]$, $\forall \ell$, is computed as
    % \begin{align}
    %   & \bbV\big[g_{\ell,k}g_{\ell,k}^*\big] \\
    %   &= \bbV\big[|\xi_{\ell,k}|^2 \mathds{1}_{\{\ell \in \mathcal{P}_k\}}\big] \\
    %     &= \left(\bbE\big[|\xi_{\ell,k}|^2 \big]\right)^2\bbV\left[\mathds{1}_{\{\ell \in \mathcal{P}_k\}}\right] + \left(\bbE\big[\mathds{1}_{\{\ell \in \mathcal{P}_k\}}\big]\right)^2\bbV\left[|\xi_{\ell,k}|^2\right] + \bbV\left[|\xi_{\ell,k}|^2\right]\bbV\left[\mathds{1}_{\{\ell \in \mathcal{P}_k\}}\right]\\
    %     & = \frac{L}{N_{\rm RF}}\left(1- \frac{L}{N_{\rm RF}}\right) + \left(\frac{L}{N_{\rm RF}}\right)^2 + \frac{L}{N_{\rm RF}}\left(1- \frac{L}{N_{\rm RF}}\right)
    %     % & \neq \frac{L}{N_{\rm RF}}. ?? \text{ \bf (should be!)} 
    \begin{align}\nonumber
       \bbV\big[g_{\ell_1,k}^*g_{\ell_2,k}\big] &= \bbV\Big[\xi_{\ell_1,k}^* \xi_{\ell_2,k} \mathds{1}_{\{\ell_1 \in \mathcal{P}_k\}}\mathds{1}_{\{\ell_2 \in \mathcal{P}_k\}}\Big] \\ \nonumber
        &\overset{(a)}{=} \bbE\big[|\xi_{\ell_1,k}^*\xi_{\ell_2,k} |^2\big]\bbE\left[\mathds{1}_{\{\ell_1,\ell_2\in \mathcal{P}_k\}}\right] \\
        \nonumber
        &~\quad -\Big(\bbE\left[\xi_{\ell_1,k}^*\xi_{\ell_2,k}\right]\Big)^2 \Big(\bbE\big[\mathds{1}_{\{\ell_1,\ell_2 \in \mathcal{P}_k\}}\big]\Big)^2 \\ \nonumber
        & = \frac{L(L-1)}{N_{\rm RF}(N_{\rm RF}-1)},
    \end{align}
where $(a)$ holds by $\bbV[XY]=\bbE[X^2]\bbE[Y^2] - (\bbE[X])^2(\bbE[Y])^2$ for independent $X$ and $Y$. Therefore, \eqref{eq:Vgw} is derived as
    \begin{align}
        \label{eq:Vgw2}
        \bbV\!\left[ \left|\bg_{k}^H \bw_i\right|^2 \right] \!=\! \frac{1}{N_{\rm RF}^2}\!\left(\!L \!+\! \sum_{\ell_1\neq \ell_2}^{N_{\rm RF}}\!\frac{L(L\!-\!1)}{N_{\rm RF}(N_{\rm RF}\!-\!1)}\!\right)\! =\! \left(\frac{L}{N_{\rm RF}}\right)^{\!2}\!.
    \end{align}
    Putting  \eqref{eq:Egw} and \eqref{eq:Vgw2} into \eqref{eq:auto_QN2}, the auto quantizaiton noise variance $\bbE\big[\Psi_k^{\rm auto} \big]$ becomes \eqref{eq:QN_auto_result}.
    % \begin{align}
    %     % \label{eq:QN_auto_result}
    %     \bbE\left[ \sum_{i=1}^{N_{\rm RF}}\left|\bh_{{\rm b},k}^H \bw_i\right|^4 \right] = \left(\frac{N_r}{L}\right)^2\sum_{i=1}^{N_{\rm RF}}\left(\left(\frac{L}{N_{\rm RF}}\right)^2+\left(\frac{L}{N_{\rm RF}}\right)^2\right) = \frac{2N_r^2}{N_{\rm RF}}.
    % \end{align}
    \qed
    
%%%%%%%%%%%%%%%%%%%%%%%%%%%%
\section{Proof of Lemma \ref{lem:QN_cross}}
\label{appx:QN_cross}
%%%%%%%%%%%%%%%%%%%%%%%%%%%%

% Now, let $\bD_u$ be a diagonal matrix whose $\ell$th diagonal element is $\mathds{1}_{\{\ell \in \mathcal{P}_u\}}$, i.e., $\bD_u$ has $L$ ones and $(N_{\rm RF} -L)$ zeros at its diagonal entries.
    We derive the cross quantization noise variance in \eqref{eq:QN_cross} as
    \begin{align} \nonumber
        \bbE\big[\Psi_k^{\rm cross}\big]\! &= \!\bbE\bigg[\sum_{i=1}^{N_{\rm RF}}\sum_{u \neq 1}^{N_u}\bh_{{\rm b},k}^H \bw_i \bw_i^H \bh_{{\rm b},u}\bh_{{\rm b},u}^H\bw_i\bw_i^H \bh_{{\rm b},k}\bigg] \\ \nonumber
        &=\! \left(\!\frac{N_r}{L}\!\right)^{\!2}\!\!\bbE_{\bg_{k}}\!\!\left[ \!\sum_{i=1}^{N_{\rm RF}}\!\sum_{u \neq 1}^{N_u}\!\bg_{k}^H\! \bw_i \!\bw_i^H \bbE_{\bg_{u}}\!\Big[\bg_{u}\bg_{u}^H\!\Big]\bw_i\!\bw_i^H \!\bg_{k}\!\right] \\ \nonumber
%        & = \left(\!\frac{N_r}{L}\!\right)^{\!2}\bbE_{\bg_{k}}
%        \!\!\left[ \sum_{i=1}^{N_{\rm RF}}\sum_{u \neq 1}^{N_u}\bg_{k}^H \bw_i \bw_i^H \bD_u\bw_i\bw_i^H \bg_{k}\right] \\ \nonumber
        & =\frac{N_r^2(N_u-1)}{LN_{\rm RF}}\sum_{i=1}^{N_{\rm RF}}\bbE_{\bg_{k}}\Big[\bg_{k}^H \bw_i \bw_i^H \bg_{k}\Big]\\ \nonumber
        % \label{eq:QN_cross_result}  
        & \stackrel{(a)}= \frac{N_r^2(N_u-1)}{N_{\rm RF}}
    \end{align}
    where $(a)$ follows from $\bbE\big[ |\bg_{k}^H \bw_i|^2 \big] = \frac{L}{{N_{\rm RF}}}$ in \eqref{eq:Egw}. 
\qed

%%%%%%%%%%%%%%%%%%%%%%%%%%
\section{Proof of Theorem \ref{thm:ergodic_rate_mrc}}
\label{appx:ergodic_rate_mrc}
%%%%%%%%%%%%%%%%%%%%%%%%%

    % The equality in \eqref{eq:E_rk_upper} holds when user beamdomain channels gains are unique in the AoA.
    To compute \eqref{eq:E_rk_approx}, we first derive $\bbE[\|\bh_{{\rm b},k}\|^2]$ as
    \begin{align}
        \label{eq:Eh}
        \bbE\Big[\|\bh_{{\rm b},k}\|^2\Big] = \frac{N_r}{L}\bbE\big[\|\bg_{k}\|^2\big] \stackrel{(a)}=N_r
    \end{align}
    where $(a)$ follows from $\|\bg_k\|^2 \sim \chi^2_{2L}$.
    Next, we compute $\bbE[\|\bh_{{\rm b},k}\|^4]$ as 
    \begin{align} \nonumber
        \bbE\Big[\|\bh_{{\rm b},k}\|^4\Big] &= \bbV\big[\|\bh_{{\rm b},k}\|^2\big] +\big(\bbE\big[\|\bh_{{\rm b},k}\|^2\big]\big)^2\\ \nonumber
        &= \left(\frac{N_r}{L}\right)^2\left(\bbV\big[\|\bg_{k}\|^2\big] +\Big(\bbE\big[\|\bg_{k}\|^2\big]\Big)^2\right) \\
        \label{eq:Ehh}
        &= \frac{N_r^2(1 + L)}{L}.
    \end{align}
    The inter-user interference term $\bbE[|\bh_{{\rm b},k}^H\bh_{{\rm b},i}|^2]$ is computed as
    \begin{align}
        \nonumber
        \bbE\Big[|\bh_{{\rm b},k}^H\bh_{{\rm b},i}|^2\Big] \!&=\! \left(\!\frac{N_r}{L}\!\right)^{\!2} \!\!\bbE\big[|\bg_{k}^H\bg_{i}|^2\big] \!=\! \left(\!\frac{N_r}{L}\!\right)^{\!\!2} \sum_{\ell=1}^{N_{\rm RF}} \! \bbE\big[|g_{\ell,k}^* g_{\ell,i}|^2\big]\\
        \nonumber
        &= \left(\!\frac{N_r}{L}\!\right)^{\!2} \sum_{\ell=1}^{N_{\rm RF}} \bbE\Big[| \xi_{\ell,k}^* \mathds{1}_{\{\ell \in \mathcal{P}_k\}} \xi_{\ell,i} \mathds{1}_{\{\ell \in \mathcal{P}_i\}}|^2\Big] 
        % = \left(\frac{N_r}{L}\right)^2 \sum_{\ell=1}^{N_{\rm RF}} \bbE\big[ \mathds{1}_{\{\ell \in \mathcal{P}_k, \}} \mathds{1}_{\{\ell \in \mathcal{P}_i\}}\big] 
        \\ \label{eq:Eh_cross} 
        & = \frac{N_r^2}{N_{\rm RF}}.
    \end{align}
    Finally, we compute the quantization variance term $\bbE[\Psi_k]$ as 
    \begin{align}
    \nonumber
       \bbE\big[\Psi_k\big] &= \bbE\big[\Psi_k^{\rm auto}\big] + \bbE\big[\Psi_k^{\rm cross}\big]\\
    %   \nonumber
    %   & \!=\! \bbE\Big[\bh_{{\rm b},k}^H\bW_{\rm DFT} {\rm diag}\big\{\bW_{\rm DFT}^H\bH_{\rm b}\bH_{\rm b}^H\bW_{\rm DFT} \big\}\bW_{\rm DFT}^H\bh_{{\rm b},k} \Big]\\  
    %   \nonumber
    %   &\!=\! \bbE\!\left[\sum_{i=1}^{N_{\rm RF}} \bh_{{\rm b},k}^H \bw_i \bw_i^H \bH_{\rm b}\bH_{\rm b}^H\bw_i\bw_i^H \bh_{{\rm b},k}\right] \\ 
    %   \nonumber
    %   &\! =\! \bbE\!\left[ \sum_{i=1}^{N_{\rm RF}}\bh_{{\rm b},k}^H \bw_i \bw_i^H \bigg(\!\bh_{{\rm b},k}\bh_{{\rm b},k}^H\!+\!\sum_{u\neq k}^{N_u}\bh_{{\rm b},u}\bh_{{\rm b},u}^H\!\bigg)\bw_i\bw_i^H \bh_{{\rm b},k}\right] \\
    %   \nonumber
    %   &\! =\! \bbE\!\left[ \!\sum_{i=1}^{N_{\rm RF}}\!\left|\bh_{{\rm b},k}^H \bw_i\right|^4 \!  +\!  \sum_{i=1}^{N_{\rm RF}}\!\sum_{u \neq 1}^{N_u}\!\bh_{{\rm b},k}^H \bw_i \bw_i^H \bh_{{\rm b},u}\bh_{{\rm b},u}^H\bw_i\bw_i^H \bh_{{\rm b},k}\!\right]\\
%       & = \underbrace{\bbE\left[ \sum_{i=1}^{N_{\rm RF}}\left|\bh_{{\rm b},k}^H \bw_i\right|^4 \right]}_{ \text{Auto quantization noise}}  + \underbrace{\bbE\left[ \sum_{i=1}^{N_{\rm RF}}\sum_{u \neq 1}^{N_u}\bh_{{\rm b},k}^H \bw_i \bw_i^H \bh_{{\rm b},u}\bh_{{\rm b},u}^H\bw_i\bw_i^H \bh_{{\rm b},k}\right]}_{\text{ Cross quantization noise}}\\
    \label{eq:QN_final}
    & \stackrel{(a)}= \frac{2N_r^2}{N_{\rm RF}} + \frac{N_r^2(N_u-1)}{N_{\rm RF}},
    \end{align}
    where $\bbE\big[\Psi_k^{\rm auto}\big]$ and $\bbE\big[\Psi_k^{\rm cross}\big]$ are in \eqref{eq:QN_auto} and \eqref{eq:QN_cross}, respectively, and $(a)$ follows from Lemma \ref{lem:QN_auto} and Lemma \ref{lem:QN_cross}.
    
    Putting \eqref{eq:Eh}, \eqref{eq:Ehh}, \eqref{eq:Eh_cross}, and \eqref{eq:QN_final} into \eqref{eq:E_rk_approx}, we derive the approximated ergodic rate of \eqref{eq:E_rk_approx} in closed form.
    The ergodic rate is equivalent to $N_u$ users, which leads to the ergodic sum rate in \eqref{eq:ER_mrc}.
    This completes the proof of Theorem \ref{thm:ergodic_rate_mrc}.
    \qed
    
%%%%%%%%%%%%%%%%%%%%%%%%%%
\section{Proof of Corollary \ref{cor:ergodic_rate_one_mrc}}
\label{appx:ergodic_rate_one_mrc}
%%%%%%%%%%%%%%%%%%%%%%%%%

Without the second analog combiner $\bW_{\rm RF}$, the approximated ergodic rate of user $k$ can be computed as \eqref{eq:E_rk_approx} by substituting the average quantization noise variance for the two-stage analog combining $\bbE[\Psi_k]$ with the following average quantization noise variance:
    \begin{align}\nonumber
        \bbE\big[\hat{\Psi}_k\big] &\!=\! \bbE\Big[\bh_{{\rm b},k}^H {\rm diag}\big\{\bH_{\rm b}\bH_{\rm b}^H\big\}\bh_{{\rm b},k}\Big]\\ \nonumber
        &\!=\! \bbE\!\left[\left(\frac{N_r}{L}\right)^2\sum_{\ell=1}^{N_{\rm RF}}|g_{\ell,k}|^2\sum_{u=1}^{N_u}|g_{\ell,u}|^2\right]\\
        % &= \bbE\left[\left(\frac{N_r}{L}\right)^2\sum_{\ell=1}^{N_{\rm RF}}|g_{\ell,k}|^2\left(|g_{\ell,k}|^2+\sum_{u\neq k}^{N_u}|g_{\ell,u}|^2\right)\right]\\
        \label{eq:QN_one}
        & \!=\! \left(\!\frac{N_r}{L}\!\right)^{\!2}\!\left(\sum_{\ell=1}^{N_{\rm RF}}\!\bbE\left[|g_{\ell,k}|^4\right]\!+\!\sum_{\ell=1}^{N_{\rm RF}}\!\sum_{u\neq k}^{N_u}\!\bbE\Big[|g_{\ell,k}|^2|g_{\ell,u}|^2\Big]\!\right).
        % & = \left(\frac{N_r}{L}\right)^2\left(\sum_{\ell=1}^{N_{\rm RF}}\bbE\left[|g_{\ell,k}|^4\right]+\sum_{u\neq k}^{N_u}\bbE\Big[|g_{\ell,k}|^2|g_{\ell,u}|^2\Big]\right).
    \end{align}
    Here, $\bbE[|g_{\ell,k}|^4]$ in \eqref{eq:QN_one} is computed as
    \begin{align}
     	\nonumber
        \bbE\Big[|g_{\ell,k}|^4\Big] &=\bbE\Big[\mathds{1}_{\{\ell \in \mathcal{P}_k\}}\Big] \bbE\Big[\big|\xi_{\ell,k}\big|^4\Big]\\
        % = \frac{L}{N_{\rm RF}}\bbE\Big[\big|\xi_{\ell,k}\big|^4\Big]\\
        \nonumber
        &=\frac{L}{N_{\rm RF}}\left(\bbV\Big[\big|\xi_{\ell,k}\big|^2\Big]+\left(\bbE\Big[\big|\xi_{\ell,k}\big|^2\Big]\right)^2\right)\\
        \label{eq:QN_one_auto}
        &= \frac{2L}{N_{\rm RF}},
    \end{align}
    and the second expectation term $\bbE[|g_{\ell,k}|^2|g_{\ell,u}|^2]$ is derived as
    \begin{align}
        \nonumber
        \bbE\Big[|g_{\ell,k}|^2|g_{\ell,u}|^2\Big] &= \bbE\Big[\mathds{1}_{\{\ell \in \mathcal{P}_k\}}\mathds{1}_{\{\ell \in \mathcal{P}_u\}}\Big] \bbE\Big[|\xi_{\ell,k}|^2|\xi_{\ell,u}|^2\Big] \\
        \label{eq:QN_one_cross}
        &= \left(\frac{L}{N_{\rm RF}}\right)^2.
    \end{align}
    Putting \eqref{eq:QN_one_auto} and \eqref{eq:QN_one_cross} into \eqref{eq:QN_one}, we derive the average quantization noise variance for the one-stage analog combining as
    \begin{align}
        \nonumber
        \bbE\big[\hat{\Psi}_k\big] = 
        % \left(\frac{N_r}{L}\right)^2\left(N_{\rm RF}\frac{2L}{N_{\rm RF}}+N_{\rm RF}(N_u-1)\left(\frac{L}{N_{\rm RF}}\right)^2\right) =
        N_r^2\left(\frac{2}{L} + \frac{N_u-1}{N_{\rm RF}}\right).
    \end{align}
    This completes the proof of Corollary \ref{cor:ergodic_rate_one_mrc}.
    % This completes the proof.
    %\begin{align}
        % \bar r_{k,{\rm one}}^{\rm mrc} 
        % &= \bbE\left[\log_2\left(1 + \frac{\rho \alpha\|\bh_{{\rm b},k}\|^4}{ \rho\alpha \sum_{i \neq k}^{N_u}|\bh_{{\rm b},k}^H\bh_{{\rm b},i}|^2 + \|\bh_{{\rm b},k}\|^2 + \rho(1-\alpha)\bh_{{\rm b},k}^H {\rm diag}\big\{\bH_{\rm b}\bH_{\rm b}^H\big\}\bh_{{\rm b},k}}\right)\right] \\
        % \label{eq:E_rk_one_approx}
        % &\stackrel{(a)} \approx \log_2\left(1 + \frac{\rho \alpha\bbE\big[\|\bh_{{\rm b},k}\|^4\big]}{ \rho\alpha \sum_{i \neq k}^{N_u}\bbE\big[|\bh_{{\rm b},k}^H\bh_{{\rm b},i}|^2\big] +\bbE\big[ \|\bh_{{\rm b},k}\|^2\big] + \rho(1-\alpha)\bbE\big[{\hat\Psi}_k\big]}\right)
    % \end{align}
    \qed
\end{appendices}

\bibliographystyle{IEEEtran}
\bibliography{hybrid2LR.bib}
\end{document}